\documentclass[12pt]{amsart}

\usepackage{etoolbox}

\newtoggle{chapters}
\togglefalse{chapters}

\usepackage[stable]{footmisc}
\usepackage[latin1]{inputenc}
\usepackage{graphicx}
\usepackage[shortlabels]{enumitem} 
\usepackage{cite}
\usepackage{amssymb,amsmath}
\usepackage{amsthm}
\usepackage{multirow}
\usepackage{hyperref}
\usepackage[margin=3cm]{geometry}
\usepackage{verbatim}
\usepackage[noend]{algpseudocode}
\usepackage{chngcntr}
\usepackage{booktabs}
\usepackage{multirow}
\usepackage{mathtools}
\usepackage{algorithm}

\usepackage{tikz}
\usetikzlibrary{arrows,backgrounds}
\usepackage[all]{xy}

\newcommand{\K}{{\mathbb K}}
\newcommand{\F}{{\mathbb F}}

\newcommand{\Q}{{\mathbb Q}}
\newcommand{\C}{{\mathbb C}}

\newcommand{\elim}[2]{{\mathrm {elim}_{#1} (#2)}}
\newcommand{\Coeff}[2]{{\mathrm {coeff}_{#1} (#2)}}
\newcommand{\I}{{\mathbf I}}
\newcommand{\V}{{\mathbf V}}
\newcommand{\initial}{\mathrm{in}}
\newcommand{\lm}{\mathrm{lm}}
\newcommand{\adj}{\mathit{adj}}
\newcommand{\equal}{\mathrel{\overset{\makebox[0pt]{\mbox{\normalfont\tiny\sffamily rad}}}{=}}}
\renewcommand{\subset}{\subseteq}
\renewcommand{\supset}{\supseteq}

\newtheorem{thm}{Theorem}
\newtheorem{prop}[thm]{Proposition}
\newtheorem{cor}[thm]{Corollary}
\newtheorem{lem}[thm]{Lemma}

\newtheorem*{ques}{Question}
\newtheorem*{notation}{Notation}
\theoremstyle{definition}
\newtheorem{defn}{Definition}[section]
\newtheorem{exmp}{Example}[section]
\theoremstyle{remark}
\newtheorem*{rem}{Remark}

\let\oldtabular\tabular
\renewcommand{\tabular}{\footnotesize\oldtabular}

\graphicspath{ {images/} }

\title
[Chordal structure in polynomial ideals]
{Exploiting chordal structure in polynomial ideals:
\\ A Gr\"obner bases approach}
\date{\today}

\author{Diego Cifuentes} 
\address{
Laboratory for Information and Decision Systems (LIDS), 
Massachusetts Institute of Technology, Cambridge MA 02139, USA}
\email{diegcif@mit.edu}

\author{Pablo A. Parrilo}
\address{
Laboratory for Information and Decision Systems (LIDS), 
Massachusetts Institute of Technology, Cambridge MA 02139, USA}
\email{parrilo@mit.edu}


\begin{document}

\begin{abstract}

Chordal structure and bounded treewidth allow for efficient
computation in numerical linear algebra, graphical models, constraint
satisfaction and many other areas. In this paper, we begin the study
of how to exploit chordal structure in computational algebraic
geometry, and in particular, for solving polynomial systems.

The structure of a system of polynomial equations can be described in
terms of a graph. By carefully exploiting the properties of this graph
(in particular, its chordal completions), more efficient algorithms
can be developed. To this end, we develop a new technique, which we
refer to as \emph{chordal elimination}, that relies on elimination
theory and Gr\"obner bases.  By maintaining graph structure throughout
the process, chordal elimination can outperform standard Gr\"obner
bases algorithms in many cases. The reason is that all computations
are done on ``smaller'' rings, of size equal to the treewidth of the
graph (instead of the total number of variables).  In particular, for
a restricted class of ideals, the computational complexity is linear
in the number of variables.  Chordal structure arises in many relevant
applications. We demonstrate the suitability of our methods in
examples from graph colorings, cryptography, sensor localization and
differential equations.

\end{abstract}

\maketitle

\section{Introduction}

Systems of polynomial equations can be used to model a large variety of applications.
In most cases the systems arising have a particular sparsity structure, and exploiting such structure can greatly improve their efficiency.
When all polynomials have degree one, we have the special case of systems of linear equations, which are often represented using matrices.
In such case, it is well known that under a \emph{chordal structure} many matrix algorithms can be done efficiently~\cite{rose1970triangulated,pothen2004elimination,paulsen1989schur}.
Similarly, many hard combinatorial problems can be solved efficiently for chordal graphs~\cite{golumbic2004algorithmic}.
Chordal graphs are also a keystone in constraint satisfaction, graphical models and database theory~\cite{Dechter2003,Beeri1983,Lauritzen1988}.
We address the question of whether chordality might also help  solve nonlinear equations.

It is natural to expect that the complexity of ``solving'' a system of
polynomials should depend on the underlying graph structure of the
equations.  In particular, a parameter of the graph called the
\emph{treewidth} determines the complexity of solving the problems
described above, and it should influence polynomial equations as well.
For instance, several combinatorial problems (e.g., Hamiltonian
circuit, vertex colorings, vertex cover) are NP-hard in general, but
are tractable if the treewidth is
bounded~\cite{bodlaender2008combinatorial}.  Nevertheless, standard
algebraic geometry techniques typically do not make use of this graph.
This paper links Gr\"obner bases with this graph structure of the
system.

It should be mentioned that, unlike classical graph problems, the
ubiquity of systems of polynomials makes them hard to solve in the
general case, even for small treewidth.  Indeed, solving zero
dimensional quadratic equations of treewidth~1 is already NP-complete,
as seen in the following example.

\begin{exmp}[Polynomials on trees are hard]\label{exmp:treewidth2hard}
Let $a_1,\ldots,a_n$ and $S$ be given integers. The Subset Sum problem
asks for a subset $A\subset \{a_1,\ldots,a_n\}$ whose sum is equal to
$S$.  Let $s_i$ denote the sum of $A\cap \{a_1,\ldots,a_i\}$ and note
that we can recover the subset $A$ from the values of
$s_1,\ldots,s_n$.  We can thus formulate the problem as:
\begin{align*}
  0 &= s_0 \\
  0 &= (s_{i} - s_{i-1})( s_i - s_{i-1} - a_i), & \mbox{for }1\leq i\leq n \\
  S &= s_n
\end{align*}
Observe that the structure of these equations can be represented with
the path graph $s_0$---$s_1$---$s_2 \cdots s_{n-1}$---$s_n$, which is
a tree. However, it is well-known that the Subset Sum problem is
NP-complete.
\end{exmp}

Despite this hardness result, it is still desirable to take advantage of this chordal structure.
In this paper, we introduce a new method that exploits this structure.
We refer to it as \emph{chordal elimination}.
Chordal elimination is based on ideas used in sparse linear algebra.
In particular, if the equations are linear chordal elimination defaults to sparse Gaussian elimination.

We proceed to formalize our statements now.
We consider the polynomial ring $R = \K[x_0,x_1,\ldots,x_{n-1}]$ over some algebraically closed field $\K$.
 We fix once and for all the lexicographic term order with $x_0>x_1>\cdots>x_{n-1}$~\footnote{Observe that smaller indices correspond to larger variables.}. 
 Given a system of polynomials $F = \{f_1,f_2,\ldots,f_s\}$ in the ring $R$, we associate to it a graph $G(F)$ with vertex set $V=\{x_0,\ldots,x_{n-1}\}$.
 Note that the vertices of $G(F)$ inherit the order from $R$.
 Such graph is given by a union of cliques: for each $f_i$ we form a clique in all its variables.
 Equivalently, there is an edge between $x_i$ and $x_j$ if and only if there is some polynomial that contains both variables. 
 We say that $G(F)$ constitutes the \emph{sparsity structure} of $F$.
 In constraint satisfaction problems, $G(F)$ is usually called the primal constraint graph~\cite{Dechter2003}.
 
 Throughout this document we fix an ideal $I\subset R$ with a given set of generators $F$.
 We assume that the associated graph $G = G(F)$ is chordal.
 %Throughout this document we fix an ideal $I\subset R$ with a given set of generators $F$.
 %We associate to $I$ the graph $G(I):=G(F)$, which we assume to be a chordal graph.
 Even more, we assume that $x_0>\cdots>x_{n-1}$ is a perfect elimination ordering (see Definition~\ref{defn:perfectelimination}) of the graph.
 In the event that $G$ is not chordal, the same reasoning applies by considering a chordal completion.
 We want to compute the elimination ideals of $I$, denoted as $\elim{l}{I}$, while preserving the sparsity structure.
 As we are mainly interested in the zero set of $I$ rather than finding the exact elimination ideals, we attempt to find some $I_l$ such that $\V(I_l)=\V(\elim{l}{I})$.

\begin{ques}\label{ques:chordal}
  Consider an ideal $I\subset R$ with generators $F$, and fix the lex order $x_0>x_1>\cdots>x_{n-1}$.
  Assume that such order is a perfect elimination ordering of its associated graph $G(F)$.
 Can we find ideals $I_l$, with some generators $F_l$, such that $\V(I_l) = \V(\elim{l}{I})$ and the sparsity structure is preserved, i.e., $G(F_l)\subset G(F)$? 
\end{ques}

We could also ask a stronger question: 
Does there exist a Gr\"obner basis $gb$ that preserves the sparsity structure, i.e., $G(gb)\subset G(F)$?
It turns out that it is not generally possible to find a Gr\"obner basis that preserves chordality, as seen in the next example.

\begin{exmp}[Gr\"obner bases may destroy chordality]\label{exmp:destroychordal}
Let $I = \langle x_0x_2 -1, x_1x_2-1\rangle$, whose associated graph is the path $x_0$---$x_2$---$x_1$. 
Note that any Gr\"obner basis must contain the polynomial $p = x_0-x_1$, breaking the sparsity structure. 
Nevertheless, we can find some generators for its first elimination ideal $\elim{1}{I} = \langle x_1x_2-1 \rangle$, that preserve such structure.
\end{exmp}

As evidenced in Example~\ref{exmp:destroychordal}, a Gr\"obner basis with the same graph structure might not exist, but we might still be able to find its elimination ideals. 
Our main method, \emph{chordal elimination}, attempts to find ideals $I_l$ as proposed above.
It generalizes the ideas of sparse linear algebra.
As opposed to Gaussian elimination, in the general case chordal elimination may not lead to the right elimination ideals.
Nevertheless, we can certify when the ideals found are correct.
This allows us to prove that for a large family of problems, which includes the case of linear equations, chordal elimination \emph{succeeds} in finding the elimination ideals.

The aim of chordal elimination is to obtain a good description of the ideal (e.g., a Gr\"obner basis), while preserving at the same time the underlying graph structure.
However, as illustrated above, there may not be a Gr\"obner basis that preserves the structure.
For larger systems, Gr\"obner bases can be extremely big and thus they may not be practical.
Nonetheless, we can ask for some sparse generators of the ideal that are the closest to such Gr\"obner basis.
We argue that one such representation can be given by finding the elimination ideals of all \emph{maximal cliques} of the graph.
We extend chordal elimination to compute these ideals in Algorithm~\ref{alg:elimcliques}.
In case $I$ is zero dimensional, it is straightforward to obtain the roots from such representation.

Chordal elimination shares many of the limitations of other elimination methods.
In particular, if $\V(I)$ is finite, the complexity depends intrinsically on the size of the projection $|\pi_l(\V(I))|$.
As such, it performs much better if such set is small.
In Theorem~\ref{thm:complexityeliml} we show complexity bounds for certain family of ideals where this condition is met.
Specifically, we show that chordal elimination is $O(n)$ if the \emph{treewidth} is bounded.

Chordal structure arises in many different applications and we believe
that algebraic geometry algorithms should take advantage of it.  The
last part of this paper evaluates our methods on some of such
applications, including cryptography, sensor localization and
differential equations.

We now summarize our contributions.
\begin{itemize}
  \item We present a new elimination algorithm that exploits chordal
    structure in systems of polynomial equations.  This method is
    presented in Algorithm~\ref{alg:eliml}.  To our knowledge, this is
    the first work that exploits chordal structure in computational
    algebraic geometry, as we will argue in the ``Related work''
    section below.

  \item We prove that the chordal elimination algorithm computes the
    correct elimination ideals for a large family of problems,
    although (as explained in Section~\ref{s:chordelim}) in general it
    may fail to do so.  In particular, Lemma~\ref{thm:exactelim}
    specifies conditions under which chordal elimination succeeds.  We
    show in Theorem~\ref{thm:exactsimplicial} that these conditions
    are met for a large class of problems.  Among others, this class
    includes linear equations and generic dense ideals.

  \item We present a recursive method
    (Algorithm~\ref{alg:elimcliques}) to compute the elimination
    ideals of all maximal cliques of the graph.  These ideals provide
    a good sparse description from which we can easily find all
    solutions, as seen in Section~\ref{s:elimcliques}.  We show in
    Corollary~\ref{thm:elimcliques} that this algorithm succeeds under
    the same conditions of chordal elimination.
  
\item We show in Theorem~\ref{thm:complexityeliml} and
  Corollary~\ref{thm:complexityvariety} that the complexity of our
  methods is linear in the number of variables and exponential in the
  treewidth for a restricted class of problems.

  \item Section~\ref{s:applications} provides experimental evaluation
    of our methods in problems from graph colorings, cryptography,
    sensor localization and differential equations.  In all these
    cases we show the advantages of chordal elimination over standard
    Gr\"obner bases algorithms.  In some cases, we show that we can
    also find a lex Gr\"obner basis faster than with degrevlex
    ordering, by making use of chordal elimination.  This subverts the
    heuristic of preferring degrevlex.
\end{itemize}

The document is structured as follows.
In Section~\ref{s:preliminaries} we provide a brief introduction to chordal graphs and we recall some ideas from algebraic geometry.
In Section~\ref{s:chordelim} we present our main method, chordal elimination.
Section~\ref{s:exactelim} presents some types of systems under which chordal elimination succeeds.
In Section~\ref{s:cliqueselim}, we present a method to find the elimination ideals of all maximal cliques of the graph. 
In Section~\ref{s:finitefield} we analyze the computational complexity of the algorithms proposed for a certain class of problems.
Finally, Section~\ref{s:applications} presents an experimental evaluation of our methods.

\subsection*{Related work}

Even though there is a broad literature regarding chordality/treewidth and also polynomial system solving, their interaction has received almost
no attention.  A meeting point between these two areas is the case of
linear equations, for which graph modelling methods have been very
successful.  There is also a lot of research studying connections
between graph theory and computational/commutative algebra.  We now
proceed to review previous works in these areas, comparing them with
our methods.

\subsubsection*{Chordality and bounded treewidth}
The concepts of chordality and bounded treewidth are pervasive in many different research areas.
In fact, several hard graph problems (e.g., vertex colorings, vertex covers, weighted independent set) can be solved efficiently in chordal graphs and in graphs of bounded treewidth~\cite{golumbic2004algorithmic,bodlaender2008combinatorial}.
In a similar way, many problems in constraint satisfaction and graphical models become polynomial-time solvable under bounded treewidth assumptions~\cite{Dechter2003,dalmau2002constraint}.
In other words, some hard problems are fixed-parameter-tractable when they are parametrized by the treewidth.

The logic community has also studied families of graph problems which are fixed-parameter-tractable with respect to the treewidth~\cite{courcelle2012graph}.
Makowsky and Meer applied these methods to algebraic problems such as evaluation, feasibility and positivity of polynomials~\cite{Makowsky2002}.
They show that these problems are tractable under bounded treewidth and finite domain conditions.
On the contrary, our methods do not require a discrete domain as they rely on well studied tools from computational algebraic geometry.
Moreover, their methods are not implementable due to the large underlying constants.

\subsubsection*{Chordality in linear algebra}

The use of graph theory methods in sparse linear algebra goes back at
least to the work of Parter~\cite{parter1961use}.  It was soon
realized that symmetric Gaussian elimination (Cholesky factorization)
does not introduce additional nonzero entries (i.e., no fill-in), only
in the case where the adjacency graph of the matrix is
chordal~\cite{rose1970triangulated}.  Current numerical linear algebra
methods exploit this property by first finding a small chordal
completion of this adjacency graph~\cite{pothen2004elimination}.  The
nonsymmetric case is quite more complicated, but a standard approach
is to use instead a chordal completion of the adjacency graph of
$A^TA$~\cite{davis2004column,pothen2004elimination}.
In this paper we generalize these ideas to the case of nonlinear equations.
We note that chordality is also used in several sparse matrix problems from optimization, such as matrix inversion, positive semidefinite matrix completion and Hessian evaluation~\cite{paulsen1989schur,OPT}. 

\subsubsection*{Structured polynomials}
Solving structured systems of polynomial equations is a well-studied
problem in computational algebraic geometry.  Many past techniques
make use of different types of structure.  In particular, properties
such as symmetry~\cite{Gatemann90,faugere2009solving} and
multi-homogeneous structure~\cite{faugere2011grobner} have been
exploited within the Gr\"obner basis framework.  Symmetry and
multi-homogeneous structure have also been used in homotopy
continuation methods; see e.g.,~\cite{Sommese2005}.

Sparsity in the equations has also been exploited by making use of
polytopal abstractions of the system~\cite{sturmpolytope}.  This idea
has led to faster algorithms based on homotopy
methods~\cite{huber1995polyhedral,li1997numerical}, sparse
resultants~\cite{emiris1995efficient} and more recently Gr\"obner
bases~\cite{faugere2014sparse}.  All these methods will perform
efficiently provided that certain measure of complexity of the system,
known as the BKK bound, is small.  Nonetheless, these type of methods
do not take advantage of the chordal structure we study in this paper.
Indeed, we will see that our methods may perform efficiently even when
the number of solutions, and thus the BKK bound, is very large.

A different body of methods come from the algebraic cryptanalysis
community, which considers very sparse equations over small finite
fields.  One popular approach is converting the problem into a SAT
problem and use SAT solvers~\cite{bard2007efficient}.  A different
idea is seen in~\cite{raddum2006new}, where they represent each
equation with its zero set and treat it as a constraint satisfaction
problem (CSP).  These methods implicitly exploit the graph structure
of the system as both SAT and CSP solvers can take advantage of it.
Our work, on the other hand, directly relates the graph structure with
the algebraic properties of the ideal.  In addition, our methods apply
to positive dimensional systems and arbitrary fields.

\subsubsection*{Graphs in computer algebra}
There is a long standing interaction between graph theory and computational algebra.
Indeed, several polynomial ideals have been associated to graphs in the past years~\cite{bayer1982division,Villarreal1990,DeLoera2009,Herzog2010binomial}.
One of the earliest such ideals was the coloring ideal~\cite{bayer1982division,hillar2008algebraic,DeLoera2009}, which was recently considered for the special case of chordal graphs~\cite{DeLoera2015}.
%The edge ideal is another widely studied example~\cite{Villarreal1990,Villarreal2015monomial}, from which several graph properties (e.g., connectedness, acyclicity, colorability, chordality) can be inferred.
The edge ideal is perhaps the most widely studied example~\cite{Villarreal1990,Villarreal2015monomial}, given its tight connections with simplicial complexes, and the many graph properties that can be inferred from the ideal (e.g., connectedness, acyclicity, colorability, chordality).
More recently, the related binomial edge ideals have also attracted a lot of research~\cite{Herzog2010binomial}.

These graph ideals allow to infer combinatorial properties by means of commutative/computational algebra, and they are also crucial in understanding the complexity of computational algebra problems.
However, previous work has mostly focused on structural properties of these specific families of ideals, as opposed to effective methods for general polynomials, such as the ones from sparse linear algebra. 
In contrast, our paper uses graph theoretic methods as a constructive guide to perform computations on arbitrary sparse polynomial systems.

%same section

\section{Preliminaries}\label{s:preliminaries}

\subsection{Chordal graphs}

Chordal graphs, also known as triangulated graphs, have many equivalent characterizations.
A good presentation is found in~\cite{Blair1993}.
For our purposes, we use the following definition.

\begin{defn}\label{defn:perfectelimination}
  Let $G$ be a graph with vertices $x_0,\ldots,x_{n-1}$.
  An ordering of its vertices $x_0>x_1>\dots>x_{n-1}$ is a \emph{perfect elimination ordering} if for each $x_l$ the set
 \begin{align}\label{eq:cliqueXl}
 X_l:=\{x_l\}\cup \{x_m: x_m\mbox{ is adjacent to }x_l,\; x_m<x_l\}
 \end{align}
 is such that the restriction $G|_{X_l}$ is a clique.
 A graph $G$ is \emph{chordal} if it has a  perfect elimination ordering.
\end{defn}
\begin{rem}
  Observe that lower indices correspond to larger vertices.
\end{rem}

Chordal graphs have many interesting properties.
Observe, for instance, that the number of maximal cliques is at most $n$.
The reason is that any clique should be contained in some $X_l$.
It is easy to see that trees are chordal graphs:
by successively pruning a leaf from the tree we get a perfect elimination ordering.

Given a chordal graph $G$, a perfect elimination ordering can be found in linear time.
A classic and simple algorithm to do so is \emph{Maximum Cardinality Search} (MCS) \cite{Blair1993}.
This algorithm successively selects a vertex with maximal number of neighbors among previously chosen vertices, as shown in Algorithm~\ref{alg:mcs}.
The ordering obtained is a reversed perfect elimination ordering.

\begin{algorithm}
  \caption{Maximum Cardinality Search~\cite{Blair1993}}
  \label{alg:mcs}
  \begin{algorithmic}[1]
    \Require{A chordal graph $G=(V,E)$ and an optional initial clique}
    \Ensure{A reversed perfect elimination ordering $\sigma$}
    \Procedure{MCS}{$G,\mathit{start}=\emptyset$}
    \State $\sigma:= \mathit{start}$
    \While {$|\sigma|<n$}
    \State choose $v \in V-\sigma$, that maximizes $|\adj(v)\cap \sigma|$ 
    \State append $v$ to $\sigma$
    \EndWhile
    \State \Return $\sigma$
    \EndProcedure
  \end{algorithmic}
\end{algorithm}

\begin{defn}\label{defn:chordalcompletion}
  Let $G$ be an arbitrary graph.
  We say that $\overline{G}$ is a \emph{chordal completion} of $G$, if it is chordal and $G$ is a subgraph of $\overline{G}$.
  The \emph{clique number} of $\overline{G}$ is the size of its largest clique.
  The \emph{treewidth} of $G$ is the minimum clique number of $\overline{G}$ (minus one) among all possible chordal completions.
\end{defn}

Observe that given any ordering $x_0>\cdots>x_{n-1}$ of the vertices of $G$, there is a natural chordal completion $\overline{G}$, i.e., we add edges to $G$ in such a way that each $G|_{X_l}$ is a clique.
In general, we want to find a chordal completion with a small clique number.
However, there are $n!$ possible orderings of the vertices and thus finding the best chordal completion is not simple.
Indeed, this problem is NP-hard~\cite{Arnborg1987}, but there are good heuristics and approximation algorithms~\cite{bodlaender2008combinatorial}.

\begin{exmp}
\begin{figure}[htb]
\centering
\includegraphics[scale=0.4]{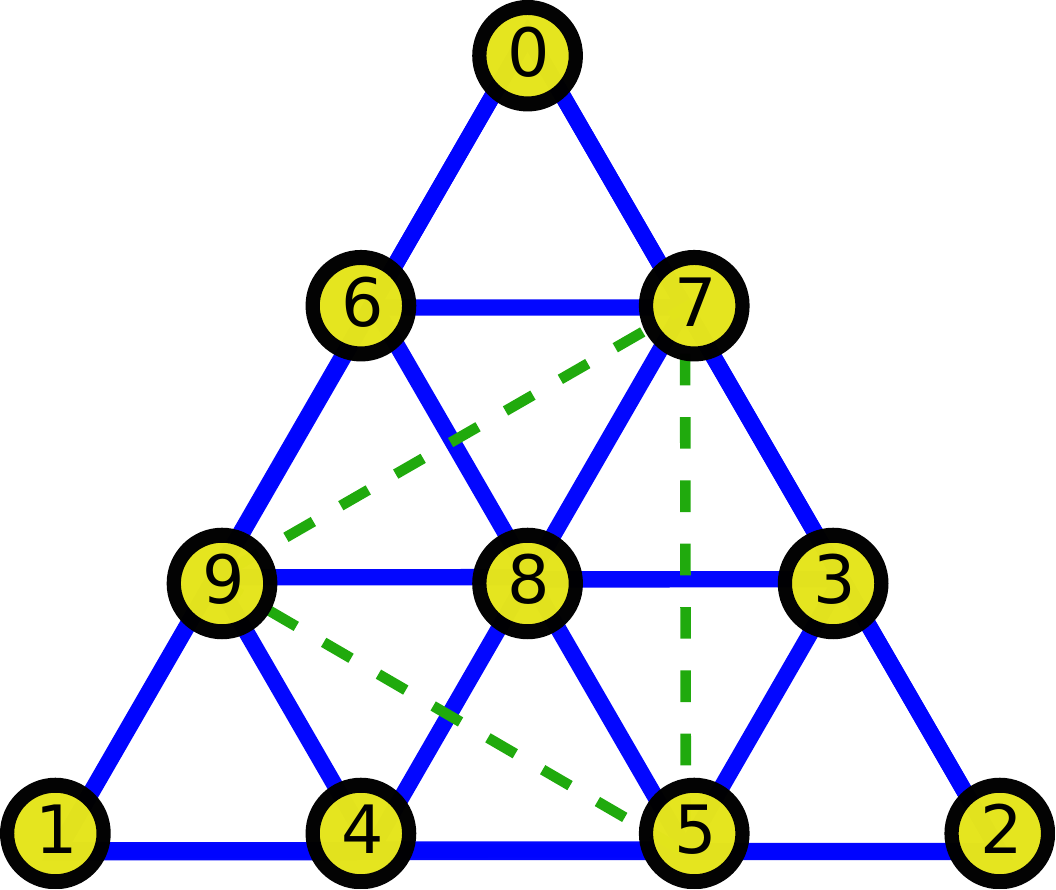}
\caption[10-vertex graph and a chordal completion]
{10-vertex graph (blue, solid) and a chordal completion (green, dashed).}
\label{fig:graph10notchordal}
\end{figure}
  Let $G$ be the blue/solid graph in Figure~\ref{fig:graph10notchordal}.
  This graph is not chordal but if we add the three green/dashed edges shown in the figure we obtain a chordal completion $\overline{G}$.
  In fact, the ordering $x_0>\cdots>x_9$ is a perfect elimination ordering of the chordal completion.
  The clique number of $\overline{G}$ is four and the treewidth of $G$ is three.
\end{exmp}

As mentioned in the introduction, we will assume throughout this document that the graph $G = G(F)$ is chordal and the ordering of its vertices (inherited from the polynomial ring) is a perfect elimination ordering.
However, for a non-chordal graph $G$ the same results hold by considering a chordal completion.

\begin{rem}[The linear case]
  We finalize this section by explaining how for linear equations finding a chordal completion of $G(F)$ agrees with standard methods from numerical linear algebra.
  A linear set of equations $F$ can be written in matrix form as $Ax = b$.
  This is equivalent to $(A^TA)x = A^Tb$, which can be solved with a Cholesky factorization.
  As mentioned earlier, to minimize the \emph{fill-in} (nonzero entries) we need to find a small chordal completion of the adjacency graph $G$ of matrix $A^TA$.
  It can be seen that this adjacency graph $G$ coincides with the graph $G(F)$ that we associate to the equations.
  Alternatively, we can directly perform an LU decomposition (Gaussian elimination) on matrix $A$, and it turns out that the adjacency graph $G$ also bounds the fill-in of the LU factors~\cite{davis2004column,pothen2004elimination}.  
\end{rem}

\subsection{Algebraic geometry}

We use standard tools from computational algebraic geometry, following the notation from~\cite{clo}.
In particular, we assume familiarity with Gr\"obner bases, elimination ideals and resultants.

We let $\elim{l}{I}$ be the $l$-th elimination ideal, i.e., 
\begin{align*}
  \elim{l}{I} := I \cap  \K[x_{l},\ldots,x_{n-1}].
\end{align*}
We will denote by $I_l$ the ``approximation'' that we will compute to this elimination ideal, defined in Section~\ref{s:chordelim}. 
We also denote $\pi_l:\K^n\to \K^{n-l}$ the projection onto the last $n-l$ coordinates.\\
We recall the correspondence between elimination and projection given by
\begin{align*}
  \V(\elim{l}{I}) = \overline{\pi_l(\V(I))}
\end{align*}
where $\overline{S}$ denotes the closure of $S$ with respect to the Zariski topology.

\begin{rem}
In order for the above equation to hold, we require that $\K$ is algebraically closed, as we assume throughout this paper.
However, if the coefficients of the equations $F$ are contained in a smaller field (e.g., $\K=\C$ but the coefficients are in $\Q$) then all computations in our algorithms will stay within such field.
\end{rem}

\section{Chordal elimination}\label{s:chordelim}

 In this section, we present our main method: chordal elimination.
 As mentioned before, we attempt to compute some generators for the elimination ideals with the same structure $G$. 
 The approach we follow mimics the Gaussian elimination process by isolating the polynomials that do not involve the variables that we are eliminating.
 The output of chordal elimination is an ``approximate'' elimination ideal that preserves chordality. 
 We call it approximate in the sense that, in general, it might not be the exact elimination ideal, but hopefully it will be close to it.
 In fact, we will find inner and outer approximations to the ideal, as will be seen later.
 In the case that both approximations are the same we will be sure that the elimination ideal was computed correctly.

 \subsection{Incremental elimination}\label{s:incrementalelim}

We follow an incremental approach to compute the elimination ideals, in a similar way as in Gaussian elimination.
We illustrate the basic methodology through the following example.

\begin{exmp}\label{exmp:elimworks}
  Consider the following ideal 
  \begin{align*}
    I = \langle x_0^4-1,x_0^2+x_2,x_1^2+x_2,x_2^2 + x_3 \rangle.
  \end{align*}
  The associated graph is the tree in Figure~\ref{fig:graph3}.
  \begin{figure}[htb]
  \centering
  \includegraphics[scale=0.35]{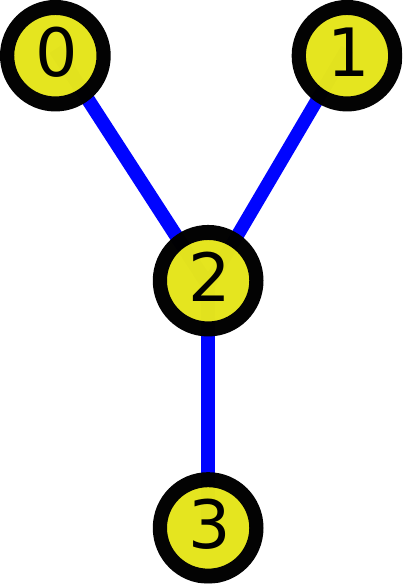}  
  \caption[Simple 3-vertex tree]
  {Simple 3-vertex tree.}
  \label{fig:graph3}
  \end{figure}
   We incrementally eliminate each of the variables, considering at each step only the equations involving it. 
   First, we consider only the polynomials involving $x_0$; there are two: $x_0^4-1,x_0^2 + x_2$.
   If we eliminate $x_0$ from these equations we obtain $x_2^2 -1$.
   This leads to the following elimination ideal
  \begin{align*}
    I_1 = \langle x_1^2+x_2,x_2^2-1,x_2^2 + x_3 \rangle.
  \end{align*}
   We now consider the polynomials involving $x_1$; there is only one: $x_1^2+x_2$. 
   Thus, we cannot eliminate $x_1$, so our second elimination ideal is
  \begin{align*}
    I_2 = \langle x_2^2-1,x_2^2 + x_3 \rangle.
  \end{align*}
   Finally, we eliminate $x_2$ from the remaining equations obtaining
  \begin{align*}
    I_3 = \langle x_3 +1 \rangle.
  \end{align*}
  For this example all elimination ideals found are correct, as can be seen from the lex Gr\"obner basis $gb = \{x_0^2+x_2,x_1^2 +x_2, x_2^2-1,x_3 + 1\}$.
\end{exmp}

 Example~\ref{exmp:elimworks} shows the basic idea we follow.
 Namely, to eliminate a variable $x_i$ we only consider a subset of the equations.
 In the above example, these equations only involved two variables at each step.
 In general, to eliminate $x_i$ we only take into account its neighboring variables in the graph.
 Therefore, if the neighborhood of each $x_i$ is small, we should require less computation.
 The chordality property will imply that these neighborhoods (cliques) are never expanded in the process.

 This successive elimination process is simple, but it is not clear whether it always leads to the correct elimination ideals.
 The following example illustrates that this is not always the case.

 \begin{exmp}[Incremental elimination may fail]\label{exmp:elimfails}
  Consider the ideal 
  $$I = \langle x_0x_1+1,x_1+x_2,x_1x_2\rangle.$$
  The associated graph is the path $x_0$---$x_1$---$x_2$.
 We proceed in an incremental way as before. 
 First, we consider only the polynomials involving $x_0$, there is only one: $x_0x_1+1$.
 Thus, we cannot eliminate $x_0$, and we are left with the ideal
 $$I_1 =  \langle x_1+x_2,x_1x_2\rangle.$$
 Eliminating $x_1$ from the two equations above, we obtain
 $$I_2 = \langle x_2^2\rangle.$$
 Observe that the original ideal $I$ is infeasible, i.e., $I= \langle 1 \rangle$, but the ideals $I_1,I_2$ found are feasible.
 Thus the elimination ideals found are not correct.
\end{exmp}

Example~\ref{exmp:elimworks} and Example~\ref{exmp:elimfails} show this incremental approach to obtain elimination ideals.
In the first case the elimination process was correct, but in the second case it was not correct.
The problem in the second example can be seen in the following equation:
 \begin{align*}
   \elim{1}{\langle x_0x_1+1,x_1+x_2,x_1x_2\rangle} \neq \elim{1}{\langle x_0x_1+1\rangle} + \langle x_1+x_2, x_1x_2 \rangle.
 \end{align*}
The goal now is to understand why these ideals are different, and when can we ensure that we successfully found the elimination ideals.

\subsection{Bounding the first elimination ideal}\label{s:squeeze}

We just introduced an incremental approach to compute elimination ideals and we observed that it might not be correct.
As will be shown next, the result of this process is always an inner approximation to the actual elimination ideal.
Even more, we will see that we can also find an outer approximation to it.
By comparing these approximations (or bounds) we can certify the cases where the elimination is correct.
We now analyze the case of the first elimination ideal, and we will later proceed to further elimination ideals.

We formalize the elimination procedure presented in Section~\ref{s:incrementalelim}.
Let $I_1$ be our estimation to the first elimination ideal as described before.
Recall that to compute the ideal $I_1$ we want to use only a subset of the equations; in the examples above the ones containing variable $x_0$.
Let's denote as $J$ the ideal of these equations, and $K$ the ideal of the remaining equations.
Then $I = J + K$, and our estimation to the first elimination ideal is given by
$I_1 = \elim{1}{J} + K$.
Note that the equations of $I$ involving $x_0$ must certainly be part of $J$.

In this way, to compute $I_1$ we only to do operations on the generators of $J$; we never deal with $K$.
As a result, the computation of $I_1$ can be done on a smaller ring, whose variables correspond to a neighborhood, or clique, of the chordal graph.
Chordality will ensure that graphical structure of $I$ is preserved, i.e., the graph associated to (the generators of) $I_1$ is a subgraph of $G$.
We elaborate more on this later.

 We want to show the relationship between our estimate $I_1$ and the actual elimination ideal $\elim{1}{I}$.
To do so, the key will be the Closure Theorem \cite[Chapter~3]{clo}.

\begin{defn}
  Let $1\leq l <n$ and let $I = \langle f_1,\ldots,f_s \rangle\subset \K[x_{l-1},\ldots,x_{n-1}]$ be an ideal with a fixed set of generators. 
  For each $1\leq t\leq s$ assume that  $f_t$ is of the form 
  \begin{align*}
    f_t = u_t(x_{l},\ldots,x_{n-1})x_{l-1}^{d_t} + (\mbox{terms with smaller degree in } x_{l-1} )
  \end{align*}
  for some $d_t\geq 0$ and $u_t\in \K[x_{l},\ldots,x_{n-1}]$.
  We define the \emph{coefficient ideal} of $I$ to be
  \begin{align*}
    \Coeff{l}{I}:= \langle u_t : 1\leq t \leq s \rangle \subset \K[x_{l},\ldots,x_{n-1}].
  \end{align*}
\end{defn}

\begin{thm}[Closure Theorem]\label{thm:closure}
  Let $I = \langle f_1,\ldots,f_s \rangle\subset \K[x_0,\ldots,x_{n-1}]$.
  Let $W := \Coeff{1}{I}$ be the coefficient ideal, let $\elim{1}{I}$ be the first elimination ideal, and let $\pi:\K^n\to \K^{n-1}$ be the projection onto the last factor.
 Then,
   \begin{align*}
     \V(\elim{1}{I}) &= \overline{\pi(\V(I))}\\
     \V(\elim{1}{I}) - \V(W) &\subset \pi(\V(I)).
   \end{align*}
\end{thm}

The next lemma tells us that $I_1$ is an inner approximation to the actual elimination ideal $\elim{1}{I}$.
It also describes an outer approximation to it, which depends on $I_1$ and some ideal $W$.
If the two bounds are equal, this implies that we successfully found the elimination ideal.

\begin{lem}\label{thm:elim}
  Let  ${J} = \langle f_1,\ldots,f_s \rangle\subset \K[x_0,\ldots, x_{n-1}]$, 
 let ${K} = \langle g_1,\ldots, g_{r}\rangle\subset \K[x_1,\ldots, x_{n-1}]$ and let
\begin{align*}
  I := {J} + {K} = \langle f_1,\ldots,f_s,g_1,\ldots,g_r \rangle.
\end{align*}
Let the ideals $I_1,W\subset \K[x_1,\ldots, x_{n-1}]$ be 
\begin{align*}
I_1 &:= \elim{1}{J} + {K}\\
W &:= \Coeff{1}{J}+ K.
\end{align*}
Then the following equations hold:
\begin{align}
  \V(\elim{1}{I}) &= \overline{\pi(\V(I))} \subset \V(I_1)\label{eq:elim1}\\
  \V(I_1) - \V(W) &\subset \pi(\V(I))\label{eq:elim2}.
\end{align}
\end{lem}

\begin{proof}
  We first show Equation~\eqref{eq:elim1}.
 The Closure Theorem says that $\V(\elim{1}{I}) = \overline{\pi(\V(I))}$.
 We will show that $\pi(\V(I))\subset \V(I_1)$,
 from which Equation~\eqref{eq:elim1} follows because $\V(I_1)$ is closed.

  In the following equations sometimes we will consider the varieties in $\K^n$ and sometimes in $\K^{n-1}$.
 To specify, we will denote them as $\V^n$ and $\V^{n-1}$ respectively.
 Notice that
  \begin{align*}
    \pi(\V^n(I)) = \pi(\V^n({J}+{K})) =  \pi(\V^{n}({J})\cap \V^n({K})).
  \end{align*}
  Now observe that
  \begin{align*}
     \pi(\V^{n}({J})\cap \V^n({K})) =  \pi(\V^{n}({J}))\cap \V^{n-1}({K}).
  \end{align*}
  The reason is the fact that if $S\subset \K^n$, $T\subset \K^{n-1}$ are arbitrary sets, then
  \begin{align*}
     \pi(S\cap (\K\times T)) =  \pi(S)\cap T.
  \end{align*}
  Finally, note that  $\overline{\pi(\V^n({J}))} =  \V(\elim{1}{J})$.
 Combining everything we conclude:
  \begin{align*}
    \pi(\V(I)) &= \pi(\V^{n}({J})) \cap \V^{n-1}({K}) \\
    &\subset \overline{\pi(\V^{n}({J}))} \cap \V^{n-1}({K}) \\
    &=  \V^{n-1}(\elim{1}{J})\cap \V^{n-1}({K}) \\
    &= \V^{n-1}(\elim{1}{J}+{K}) \\
    &= \V(I_1).
  \end{align*}

  We now show Equation~\eqref{eq:elim2}.
  The Closure Theorem states that 
  $$\V(\elim{1}{J}) - \V(\Coeff{1}{J}) \subset \pi(\V({J})).$$
 Then,
  \begin{align*}
    \V(I_1) - \V(W) &= [\V^{n-1}(\elim{1}{J})\cap \V^{n-1}({K})]  - [\V^{n-1}(\Coeff{1}{J})\cap \V^{n-1}(K)] \\
    & = [\V^{n-1}(\elim{1}{J})-\V^{n-1}(\Coeff{1}{J})]\cap \V^{n-1}({K})\\
    & \subset \pi(\V^n({J}))\cap \V^{n-1}({K})\\
    & = \pi(\V(I)).
  \end{align*}
  This concludes the proof.
\end{proof}

Note that the lemma above implies the following equations:
\begin{align*}
  \V(I_1) - \V(W) &\subset \V(\elim{1}{I}) \subset \V(I_1) \\
  \sqrt{I_1}: \sqrt{W} &\supset \sqrt{\elim{1}{I}}\supset \sqrt{I_1}
\end{align*}
where we used that set difference of varieties corresponds to ideal quotient.
Thus, the ideal $W$ bounds the approximation error of our estimation $I_1$ to the ideal $\elim{1}{I}$.
 In particular, if $\V(W)$ is empty then $I_1$ and $\elim{1}{I}$ determine the same variety.

\subsection{Bounding all elimination ideals}

Lemma~\ref{thm:elim} gave us the relationship between our estimation $I_1$ and the actual elimination ideal $\elim{1}{I}$.
We generalize this now to further elimination ideals.

We denote by $I_l$ to our estimation to the $l$-th elimination ideal.
As before, to estimate $\elim{l+1}{I_l}$ we only use a subset of the equations of $I_l$, which we denote as $J_l$.
The remaining equations are denoted as $K_{l+1}$.
Then $I_{l+1} = \elim{l+1}{I_l} + K_{l+1}$.
The following theorem establishes the relationship between $I_{l+1}$ and $\elim{l+1}{I}$.

\begin{thm}\label{thm:eliml}
  Let $I\subset \K[x_0,\ldots,x_{n-1}]$ be an ideal.
  Consider ideals $I_l\subset \K[x_l,\ldots,x_{n-1}]$ for $0\leq l< n$, with $I_0:= I$, which are constructed recursively as follows:
  \begin{enumerate}[(i)]\itemsep0em
    \item\label{line:fakedecompse} Given $I_l$, let $J_{l}\subset \K[x_{l},\ldots,x_{n-1}]$, $K_{l+1}\subset \K[x_{l+1},\ldots,x_{n-1}]$ be\footnote{Note that this decomposition is not unique, since we are not fully specifying the ideals $J_l,K_{l+1}$.}
      such that  $I_{l} = J_{l} + K_{l+1}$. 
    \item Let
      $I_{l+1} := \elim{l+1}{J_l} + K_{l+1}$.
    \item Also denote $W_{l+1} := \Coeff{l+1}{J_l}+ K_{l+1}$.
  \end{enumerate}
  Then for each $l$ the following equations hold:
\begin{align}
  \V(\elim{l}{I}) &= \overline{\pi_l(\V(I))} \subset \V(I_l)\label{eq:eliml1}\\
  \V(I_l) - [\pi_l(\V(W_1))\cup \dots\cup \pi_l(\V(W_l))] &\subset \pi_l(\V(I))\label{eq:eliml2}.
\end{align}
\end{thm}
\begin{proof}
  It follows from Lemma~\ref{thm:elim} by induction. See Appendix~\ref{s:proofschordelim}.
\end{proof}

The lemma above implies the following equations:
\begin{subequations} \label{eq:boundeliml}
\begin{align}  
  \V(I_L) - \V(W) &\subset \V(\elim{L}{I}) \subset \V(I_L)\\
  \sqrt{I_L}:\sqrt{W} &\supset \sqrt{\elim{L}{I}}\supset \sqrt{I_L}
\end{align}
\end{subequations}
where the ideal $W$ is
\begin{align}\label{eq:intersectionWl}
  W := \elim{L}{W_1}\cap \dots\cap \elim{L}{W_L}.
\end{align}
Note also that by construction we always have that if $x_m<x_l$ then $I_m\subset I_l$.

\subsection{Chordal elimination algorithm}\label{s:eliml}

The recursive construction given in Theorem~\ref{thm:eliml} is not fully specified (see item~\ref{line:fakedecompse}).
In particular, it is not clear which decomposition $I_l = J_l + K_{l+1}$ to use.
We now describe the specific decomposition we use to obtain the chordal elimination algorithm.
 
We recall the definition of the cliques $X_l$ from Equation~\eqref{eq:cliqueXl}.
Equivalently, $X_l$ is the largest clique containing $x_l$ in $G|_{\{x_l,\ldots,x_{n-1}\}}$.
Let $f_j$ be a generator of $I_l$.
If all the variables in $f_j$ are contained in $X_l$, we put $f_j$ in $J_l$.
Otherwise, if some variable of $f_j$ is not in $X_l$, we put $f_j$ in $K_{l+1}$.
We refer to this procedure as \emph{clique decomposition}.

\begin{exmp}
 Let $I = \langle f,g,h\rangle$ where  $f = x_0^2+ x_1x_2$, $g = x_1^3 + x_2$ and $h = x_1 +x_3$. 
 Note that the associated graph consists of a triangle $x_0,x_1,x_2$ and the edge $x_1,x_3$.
 Thus, we have $X_0 = \{x_0,x_1,x_2\}$.
 The clique decomposition sets $J_0 = \langle f,g\rangle$, $K_1 = \langle h\rangle$.
\end{exmp}

We should mention that this decomposition method is reminiscent to the bucket elimination algorithm from constraint satisfaction~\cite{Dechter1998}.
However, we do not place an equation $f$ in its largest variable, but rather in the largest variable $x_l$ such that $f\in \K[X_l]$.
The reason for doing this is to shrink the variety $\V(J_l)$ further.
This leads to a tighter approximation of the elimination ideals and simplifies the Gr\"obner basis computation.

It is easy to see that the procedure in Theorem~\ref{thm:eliml}, using this clique decomposition, preserves chordality. 
We state that now.

\begin{prop}\label{thm:decompositionpreserve}
  Let $I$ be an ideal with chordal graph $G$.
 If we follow the procedure in Theorem~\ref{thm:eliml} using the clique decomposition, then the graph associated to $I_l$ is a subgraph of $G$.
\end{prop}
\begin{proof}
  Observe that we do not modify the generators of $K_{l+1}$, and thus the only part where we may alter the sparsity pattern is when we compute $\elim{l+1}{J_l}$ and $\Coeff{l+1}{J_l}$.
  However, the variables involved in $J_l$ are contained in the clique $X_l$ and thus, independent of which operations we apply to its generators, we will not alter the structure.
\end{proof}

After these comments, we solved the ambiguity problem of step~\ref{line:fakedecompse} in Theorem~\ref{thm:eliml}.
However, there is still an issue regarding the ``error ideal'' $W$ of Equation~\eqref{eq:intersectionWl}.
We recall that $W_{l+1}$ depends on the coefficient ideal of $J_l$.
Thus, $W_{l+1}$ does not only depend on the ideal $J_l$, but it depends on the specific set of generators that we are using.
In particular, some set of generators might lead to a larger/worse variety $\V(W_{l+1})$ than others.
This problem is inherent to the Closure theorem, and it is discussed in \cite[Chapter~3]{clo}.
It turns out that a lex Gr\"obner basis of $J_l$ is an optimal set of generators, as shown in~\cite{clo}.
Therefore, it is convenient to find this Gr\"obner basis before computing the coefficient ideal.

Algorithm~\ref{alg:eliml} presents the chordal elimination algorithm.
The output of the algorithm is the inner approximation $I_L$ to the $L$-th elimination ideal and the ideals $W_1,\ldots,W_L$, that satisfy~\eqref{eq:boundeliml}.
In the event that $\V(W_l) = \emptyset$ for all $l$, the elimination was correct.
This is the case that we focus on for the rest of the paper.

\begin{algorithm}
  \caption{Chordal Elimination to find the $L$-th Elimination Ideal}
  \label{alg:eliml}
  \begin{algorithmic}[1]
    \Require{An ideal $I$, given by generators with chordal graph $G$, and an integer $L$}
    \Ensure{Ideals $I_L$ and $W_1,\ldots,W_L$ approximating $\elim{L}{I}$ as in~\eqref{eq:boundeliml}}
    \Procedure{ChordElim}{$I,G,L$}
    \State $I_0 = I$
    \For {$l = 0:L-1$}
    \State get clique $X_l$ of $G$
    \State $J_l,K_{l+1} = $ \Call{SplitGens}{$I_l,X_l$}
    \State \Call{FindElim\&Coeff}{$J_l$}
    \State $I_{l+1} = \elim{l+1}{J_l} + K_{l+1} $
    \State $W_{l+1} = \Coeff{l+1}{J_l}+ K_{l+1}$
    \EndFor
    %\State $W = $ ChordElim($W_1,L$) $\cap \cdots \cap$ ChordElim($W_L,L$) 
    \State \Return $I_L, W_1,\ldots,W_L$
    \EndProcedure
    \vspace{5pt}
    \Procedure{SplitGens}{$I_l,X_l$}
    \Comment{Partition generators of $I_l$}
    \State $J_l = \langle f: f$ generator of $I_l$ and $f \in\K[X_l] \rangle$
    \State $K_{l+1} = \langle f: f$ generator of $I_l$ and $f \notin\K[X_l]\rangle$
    \EndProcedure
    \vspace{5pt}
    \Procedure{FindElim\&Coeff}{$J_l$}
    \Comment{Eliminate $x_l$ in the ring $\K[X_l]$}
    \State append to $J_l$ its lex Gr\"obner basis \label{line:appendgroebner}
    \State $\elim{l+1}{J_l} = \langle f: f $ generator of $J_l$ with no $x_l\rangle$
    \State $\Coeff{l+1}{J_l} = \langle$leading coefficient of $ f: f $ generator of $J_l\rangle$
    \EndProcedure
  \end{algorithmic}
\end{algorithm}

\begin{rem}
Observe that in line~\ref{line:appendgroebner} of Algorithm~\ref{alg:eliml}  we append a Gr\"obner basis to $J_l$, so that we do not remove the old generators.
There are two reasons to compute this lex Gr\"obner basis:
   it allows to find the $\elim{l+1}{J_l}$ easily and we obtain a tighter $W_{l+1}$ as discussed above.
However, we do not replace the old set of generators but instead we append to them this Gr\"obner basis.
We will explain the reason to do that in Section~\ref{s:eliminationtree}. 
\end{rem}

\subsection{Elimination tree}\label{s:eliminationtree}

We now introduce the concept of elimination tree, and show its connection with chordal elimination.
This concept will help us to analyze our methods.

\begin{defn}\label{defn:eliminationtree}
  Let $G$ be an ordered graph with vertex set $x_0>\cdots>x_{n-1}$.
  We associate to $G$ the following \emph{directed spanning tree} $T$ that we refer to as the \emph{elimination tree}:
  For each $x_l>x_{n-1}$ there is an arc from $x_l$ towards the largest  $x_p$ that is adjacent to $x_l$ and $x_p<x_l$.
 We will say that $x_p$ is \emph{the parent} of $x_l$ and $x_l$ is \emph{a descendant} of $x_p$.
 Note that $T$ is rooted at $x_{n-1}$.
\end{defn}

\begin{figure}[htb]
\centering
\includegraphics[scale=0.4]{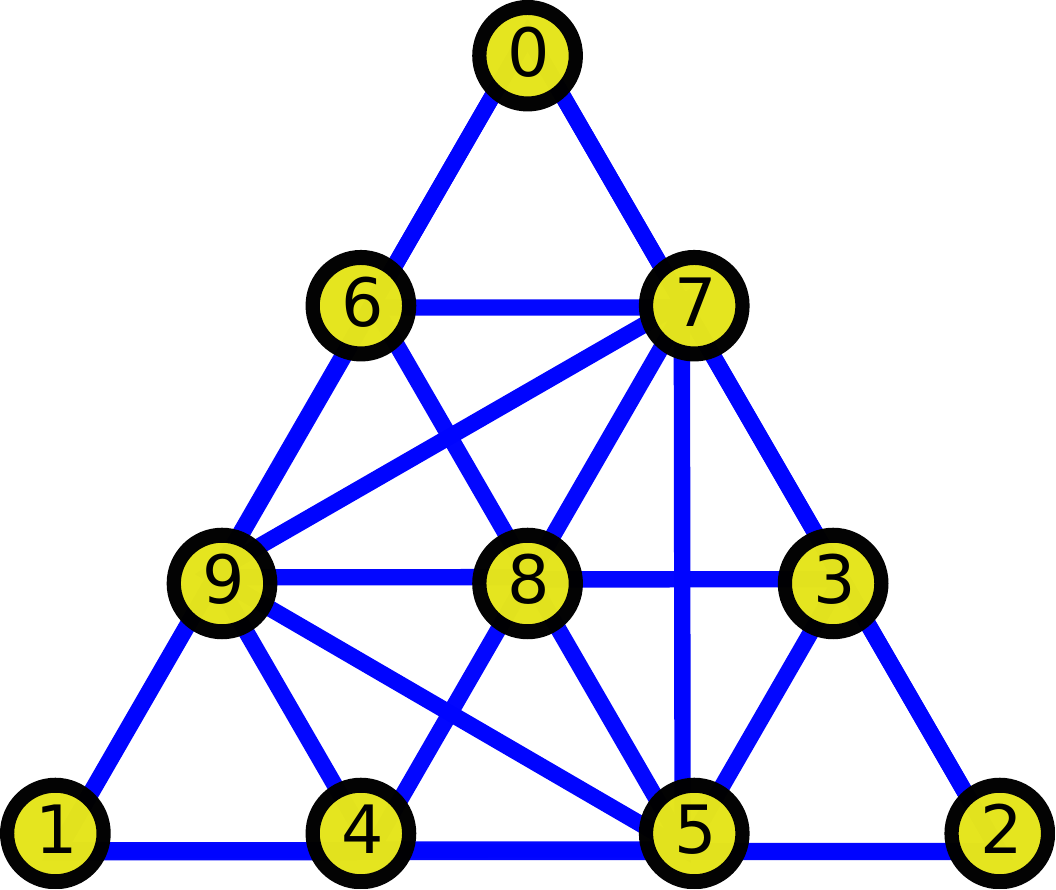}  \hspace*{10pt}
\includegraphics[scale=0.4]{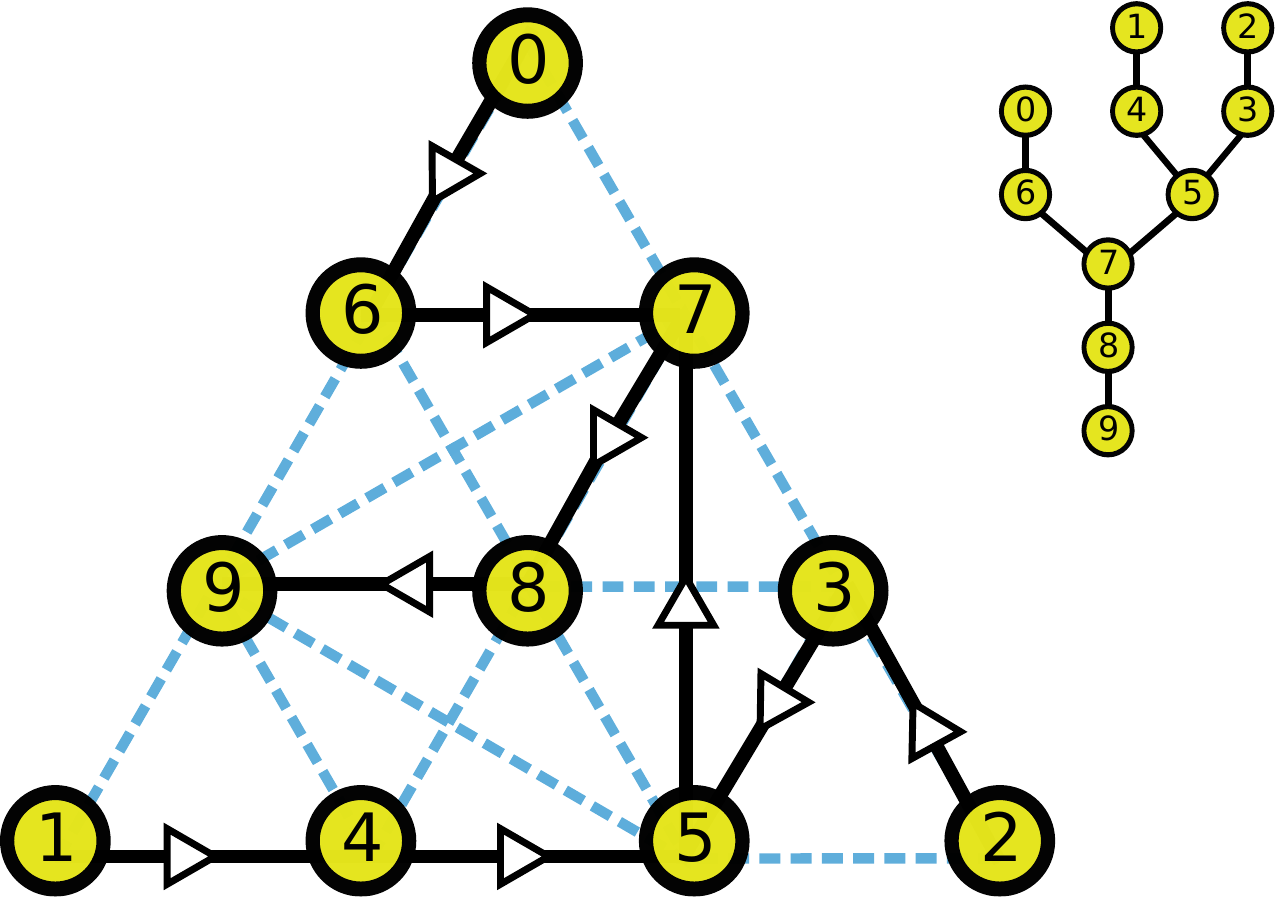}  
\caption[10-vertex chordal graph and its elimination tree]
{Chordal graph $G$ and its elimination tree $T$.}
\label{fig:graph10}
\end{figure}

Figure~\ref{fig:graph10} shows an example of the elimination tree of a given graph.
It is easy to see that eliminating a variable $x_l$ corresponds to pruning one of the leaves of the elimination tree. 
We now present a simple property of such tree.

\begin{lem}\label{thm:cliquecontainment}
  Let $G$ be a chordal graph, let $x_l$ be some vertex and let $x_p$ be its parent in the elimination tree $T$. 
  Then,
  \begin{align*}
    X_l \setminus \{x_l\}\subset X_p
  \end{align*}
  where $X_i$ is as in Equation~\eqref{eq:cliqueXl}.
\end{lem}
\begin{proof}
  Let $C =  X_l\setminus \{x_l\}$. 
  Note that $C$ is a clique that contains $x_p$. 
  Even more, $x_p$ is its largest variable because of the definition of $T$. 
  As $X_p$ is the unique largest clique satisfying such property, we must have $C\subset X_p$. 
\end{proof}

A consequence of the lemma above is the following relation:
\begin{align}\label{eq:gradeddectrivial}
  \elim{l+1}{I \cap \K[X_l]} \subset I\cap \K[X_p]
\end{align}
where $I \cap \K[X_l]$ is the set of all polynomials in $I$ that involve only variables in $X_l$.
The reason of the relation above is that
\begin{align*}
  \elim{l+1}{I \cap \K[X_l]} = (I \cap \K[X_l]) \cap \K[x_{l+1},\ldots,x_{n-1}] = I\cap \K[X_l\setminus \{x_l\}].
\end{align*}

There is a simple geometric interpretation of Equation~\eqref{eq:gradeddectrivial}.
The variety $\V(I\cap \K[X_l])$ can be interpreted as the set of partial solutions restricted to the set $X_l$.
Thus, Equation~\eqref{eq:gradeddectrivial} is telling us that any partial solution on $X_p$ extends to a partial solution on $X_l$ (the inclusion is reversed).
Even though this equation is very simple, this is a property that we would like to keep in chordal elimination. 

Clearly, we do not have a representation of the clique elimination ideal $I\cap \K[X_l]$.
However, the natural relaxation to consider is the ideal $J_l\subset \K[X_l]$ that we compute in Algorithm~\ref{alg:eliml}.
To preserve the property above (i.e., every partial solution of $X_p$ extends to $X_l$), we would like to have the following relation:
\begin{align}
  \elim{l+1}{J_l} \subset J_p.
  \label{eq:gradeddec}
\end{align}

It turns out that there is a very simple way to ensure this property: we preserve the old generators of the ideal during the elimination process.
This is precisely the reason why in line~\ref{line:appendgroebner} of Algorithm~\ref{alg:eliml} we \emph{append} a Gr\"obner basis to $J_l$.

We now prove that Equation~\eqref{eq:gradeddec} holds.
We need one lemma before.

\begin{lem}\label{thm:eliminiationcontainment}
  In Algorithm~\ref{alg:eliml}, let $f\in I_l$ be one of its generators. If $x_m$ is such that $x_m\leq x_l$ and  $f\in \K[X_m]$, then $f$ is a generator of $J_m$.
  In particular, this holds if $x_m$ is the largest variable in $f$.
\end{lem}
\begin{proof}
  For a fixed $x_m$, we will show this by induction on $x_l$.

  The base case is $l = m$. 
  In such case, by construction of $J_m$ we have that $f\in J_m$.

  Assume now that the assertion holds for any $x_l'$ with  $x_m\leq x_{l'}<x_l$ and let $f$ be a generator of $I_l$.
  There are two cases: either $f\in J_{l}$ or $f\in K_{l+1}$.
  In the second case, $f$ is a generator of $I_{l+1}$ and using of the induction hypothesis we get $f\in J_m$.
  In the first case, as $f\in \K[X_m]$ then all variables of $f$ are less or equal to $x_m$, and thus strictly smaller than $x_l$.
  Following Algorithm~\ref{alg:eliml}, we see that $f$ is a generator of $ \elim{l+1}{J_l}$.
  Thus, $f$ is again a generator of $I_{l+1}$ and we conclude by induction.

  We now prove the second part, i.e., it holds if $x_m$ is  the largest variable.
  We just need to show that $f\in \K[X_m]$.
  Let $X_{lm}:=X_l\cap \{x_m,\ldots,x_{n-1}\}$, then $f\in \K[X_{lm}]$ as $x_m$ is the largest variable.
  Note that as $f \in \K[X_l]$ and $f$ involves $x_m$, then $x_m\in X_l$.
  Thus, $X_{lm}$ is a clique of $G|_{\{x_m,\ldots,x_{n-1}\}}$ and it contains $x_m$.
  However, $X_m$ is the unique largest clique that satisfies this property.
  Then $X_{lm}\subset X_m$ so that $f\in \K[X_{lm}]\subset \K[X_m]$.
\end{proof}

\begin{cor}
  Let $x_l$ be arbitrary and let $x_p$ be its parent in the elimination tree~$T$.
  Then Equation~\eqref{eq:gradeddec} holds for Algorithm~\ref{alg:eliml}.
\end{cor}
\begin{proof}
  Let $f\in \elim{l+1}{J_l}$ be one of its generators.
  Then $f$ is also one of the generators of $I_{l+1}$, by construction.
  It is clear that the variables of $f$ are contained in $X_l\setminus \{x_l\}\subset X_p$, where we used Lemma~\ref{thm:cliquecontainment}.
  From Lemma~\ref{thm:eliminiationcontainment} we get that $f\in J_p$, concluding the proof.
\end{proof}

The reader may believe that preserving the old set of generators is not necessary.
The following example shows that it is necessary in order to have the relation in~\eqref{eq:gradeddec}.

\begin{exmp}
  Consider the ideal 
  $$I = \langle x_0-x_2,x_0-x_3,x_1-x_3,x_1-x_4,x_2-x_3,x_3-x_4,x_2^2 \rangle,$$
  whose associated graph consists of two triangles $\{x_0,x_2,x_3\}$ and $\{x_1,x_3,x_4\}$.
  Note that the parent of $x_0$ is $x_2$.
  If we preserve the old generators (as in Algorithm~\ref{alg:eliml}), we get $I_2 = \langle x_2-x_3,x_3-x_4,x_2^2,x_3^2,x_4^2\rangle$.
  If we do not preserve them, we get instead $\hat{I_2} = \langle x_2-x_3,x_3-x_4,x_4^2\rangle$.
  In the last case we have $\hat{J_2}=\langle x_2-x_3\rangle$ so that $x_2^2\notin \hat{J_2}$, even though $x_2^2 \in J_0$.
  Moreover, the ideal $J_0$ is zero dimensional, but $\hat{J_2}$ has positive dimension.
  Thus, Equation~\eqref{eq:gradeddec} does not hold.
\end{exmp}

\section{Successful elimination}\label{s:exactelim}

\begin{notation}
We will write  $I\equal J$ whenever we have $\V(I)=\V(J)$.
\end{notation}

In Section~\ref{s:chordelim} we showed an algorithm that gives us an approximate elimination ideal.
In this section we are interested in finding conditions under which such algorithm returns the actual elimination ideal.
We will say that chordal elimination, i.e., Algorithm~\ref{alg:eliml}, \emph{succeeds} if we have $\V(I_l) = \V(\elim{l}{I})$. 
Following the convention above, we write $I_l \equal \elim{l}{I}$.

\subsection{The domination condition}
 
Theorem~\ref{thm:eliml} gives us lower and upper bounds on the actual elimination ideals.
We use these bounds to obtain a condition that guarantees that chordal elimination succeeds.

\begin{defn}\label{defn:dominated_xi}
   We say that a polynomial $f$ is \emph{$x_i$-dominated} if its leading monomial has the form $x_i^{d}$ for some $d$.
   We say that an ideal $J$ is $x_i$-dominated if there is some $f\in J$ that is $x_i$-dominated. 
   Equivalently, $J$ is $x_i$-dominated if its initial ideal $\initial(J)$ contains a pure power of $x_i$.
\end{defn}

\begin{defn}[Domination condition]\label{defn:dominationcondition}
  Let $I$ be an ideal and use Algorithm~\ref{alg:eliml}.
  We say that the \emph{domination condition} holds if $J_l$ is $x_l$-dominated for each $l$. 
\end{defn}

The domination condition implies that chordal elimination succeeds, as shown now.

\begin{lem}[Domination implies Success]\label{thm:exactelim}
  If the domination condition holds then $I_l \equal \elim{l}{I}$ and the corresponding variety is $\pi_l(\V(I))$, for all $l$.
\end{lem}
\begin{proof}
  As $J_l$ is $x_l$-dominated, then its initial ideal $\initial(J_l)$ contains a pure power of $x_l$.
 Thus, there must be a $g$ that is part of the Gr\"obner basis of $J_l$ and is $x_l$-dominated.
 The coefficient $u_t$ that corresponds to such $g$ is $u_t=1$, and therefore $1\in W_l$ and $\V(W_l) = \emptyset$.
 Thus the two bounds in Theorem~\ref{thm:eliml} are the same and the result follows.
\end{proof}

We will now show some classes of ideals in which the domination condition holds.
Using the previous lemma, this guarantees that chordal elimination succeeds.

\begin{cor}\label{thm:exactzerodimensional}
  Let $I$ be an ideal  and assume  that for each $l$ such that $X_l$ is a maximal clique of $G$, the ideal $J_l\subset \K[X_l]$ is zero dimensional.
 Then the domination condition holds and chordal elimination succeeds.
\end{cor}
\begin{proof}
  Let $x_m$ be arbitrary, and let $x_l\geq x_m$ be such that $X_m\subset X_l$ and $X_l$ is a maximal clique.
  As $J_l\subset \K[X_l]$ is zero dimensional, then it is $x_j$-dominated for all $x_j\in X_l$.
 Thus, there is a $g$ that is part of the Gr\"obner basis of $J_l$ and is $x_m$-dominated.
  From Lemma~\ref{thm:eliminiationcontainment} we obtain that $g\in J_m$, and thus the domination condition holds.
\end{proof}

\begin{cor}\label{thm:exactsimplicialeachl}
Let $I$ be an ideal and assume that for each $l$ there is a generator $f_l$ of $I$ that is $x_l$-dominated.
 Then the domination condition holds and chordal elimination succeeds.
\end{cor}
\begin{proof}
  It follows from Lemma~\ref{thm:eliminiationcontainment} that $f_l \in J_l$, so that $J_l$ is $x_l$-dominated and the domination condition holds.
\end{proof}

The previous corollary presents a first class of ideals for which we are guaranteed to have successful elimination.
 Note that when we solve equations over a finite field $\F_q$, usually we include equations of the form $x_l^q-x_l$, so the corollary holds.
 In particular, it holds for $0/1$ problems. 
 
\subsection{Simplicial equations}

The assumptions of Corollary~\ref{thm:exactsimplicialeachl} are too strong for many cases.
In particular, if $l=n-2$ the only way that such assumption holds is if there is a polynomial that only involves $x_{n-2},x_{n-1}$.
We will show now a bigger class of ideals for which the domination condition also holds and thus chordal elimination succeeds.
The following concept is the basis for this class.

 \begin{defn}\label{defn:simplicial}
  Let $f\in \K[x_0,\ldots,x_{n-1}]$ be such that for each variable $x_l$ of positive degree,  the monomial $m_l$ of $f$ with largest degree in $x_l$ is unique and has the form $m_l=x_l^{d_l}$ for some $d_l>0$.
 We say that $f$ is \emph{simplicial}.
\end{defn}

\begin{exmp}
  Consider the polynomials of Example~\ref{exmp:elimfails}: 
  $$f_1 = x_0x_1+1,\qquad f_2 = x_1+x_2,\qquad f_3 = x_1x_2.$$
  Then $f_2$ is simplicial, as for both $x_1,x_2$ the monomials of largest degree in these variables are pure powers.
  In general, linear equations are always simplicial.
  On the other hand, $f_1,f_3$ are not simplicial.
  This makes sense, as we will see that if all polynomials are simplicial then chordal elimination succeeds, which was not the case of Example~\ref{exmp:elimfails}.
  On the contrary, all the polynomials of Example~\ref{exmp:elimworks} are simplicial.
\end{exmp}

Note that the definition of simplicial is independent of the monomial ordering used, as opposed to $x_i$-domination.
The reason for the term simplicial is that the (scaled) standard simplex 
$$\Delta = \{x:x\geq0,\sum_{x_l\in X_f} x_l/d_l = |X_f|\}$$
where $X_f$ are the variables of $f$, is a face of the Newton polytope of $f$ and they are the same if $f$ is homogeneous.
 
We will make an additional genericity assumption on the simplicial polynomials.
Concretely, we assume that the coefficients of $m_l=x_l^{d_l}$ are generic, in a sense that will be clear in the next lemma.

\begin{lem}\label{thm:resultants}
  Let $q_1,q_2$ be generic simplicial polynomials.
 Let $X_1,X_2$ denote their sets of variables and let $x\in X_1\cap X_2$.
  Then $h = \mathrm{Res}_{x}(q_1,q_2)$ is generic simplicial and its set of variables is $X_1\cup X_2\setminus x$.
  \end{lem}
  \begin{proof}
    Let $q_1,q_2$ be of degree $m_1,m_2$ as univariate polynomials in $x$.
 As $q_2$ is simplicial, for each $x_i\in X_2\setminus x$ the monomial with largest degree in $x_i$ has the form $x_i^{d_2}$.
 It is easy to see that the largest monomial of $h$, as a function of $x_i$, that comes from $q_2$ will be $x_i^{d_2m_1}$.
 Such monomial arises from the product of the main diagonal of the Sylvester matrix.
 In the same way, the largest monomial that comes from $q_1$ has the form $x_i^{d_1m_2}$.
  If $d_2m_1=d_1m_2$, the genericity guarantees that such monomials do not cancel each other out.
  Thus, the leading monomial of $h$ in $x_i$ has the form $x_i^{\max\{d_2m_1,d_1m_2\}}$ and then $h$ is simplicial.
 The coefficients of the extreme monomials are polynomials in the coefficients of $q_1,q_2$, so if they were not generic (they satisfy certain polynomial equation), then $q_1,q_2$ would not be generic either.
  \end{proof}

  Observe that in the lemma above we required the coefficients to be generic in order to avoid cancellations in the resultant.
 This is the only part where we need this assumption.

 We recall that elimination can be viewed as pruning the elimination tree $T$ of $G$ (Definition~\ref{defn:eliminationtree}).
 We attach each of the generators of $I$ to some node of $T$.
 More precisely, we attach a generator $f$ to the largest variable it contains, which we denote as $x(f)$.
 The following lemma tells us that if there are many simplicial polynomials attached to the subtree of $T$ rooted in $x_l$, then $J_l$ is $x_l$-dominated.

  \begin{lem}\label{thm:getsimplicial}
    Let $I=\langle f_1,\ldots,f_s\rangle$ and let $1\leq l<n$.
    Let $T_l$ be a subtree of $T$ with  $t$ vertices and minimal vertex $x_l$.
    Assume that there are $f_{i_1},\ldots,f_{i_{t}}$ generic simplicial with largest variable $x(f_{i_j})\in T_l$ for $1\leq j\leq t$.
    Then $J_l$ is $x_l$-dominated.
  \end{lem}
  \begin{proof}
    We will show that we can find a simplicial polynomial $f_l\in J_l$ that contains $x_l$, which implies the desired result.
    Let's ignore all $f_t$ such that its largest variable is not in $T_l$.
    By  doing this, we get smaller ideals $J_l$, so it does not help to prove the statement.
    Let's also ignore all vertices which do not involve one of the remaining equations. 
    Let $S$ be the set of variables which are not in $T_l$. 
    As in any of the remaining equations the largest variable should be in $T_l$, then for any $x_i\in S$ there is some $x_j\in T_l$ with $x_j>x_i$.
    We will show that for any $x_i\in S$ we have $x_l>x_i$. 
    
    Assume by contradiction that it is not true, and let $x_i$ be the smallest counterexample.
    Let $x_p$ be the parent of $x_i$. Note that $x_p\notin S$ because of the minimality of $x_i$, and thus $x_p \in T_l$.
    As mentioned earlier, there is some $x_j\in T_l$ with $x_j>x_i$.
    As $x_j>x_i$ and $x_p$ is the parent of $x_i$, this means that $x_i$ is in the path of $T$ that joins $x_j$ and $x_p$.
    However $x_j,x_p\in T_l$ and $x_i\notin T_l$ so this contradicts that $T_l$ is connected.

    Thus, for any $x_i\in S$, we have that $x_i<x_l$. 
    This says that to obtain $J_l$ we don't need to eliminate any of the variables in $S$. 
    Therefore, we can ignore all variables in $S$.
    Thus, we can assume that $l=n-1$ and $T_l=T$.
    This reduces the problem the specific case considered in the following lemma.
  \end{proof}

  \begin{lem}\label{thm:exactsquare}
    Let $I=\langle f_1,\ldots,f_{n}\rangle$ such that $f_j$ is generic simplicial for all $j$.
 Then there is a simplicial polynomial $f\in I_{n-1}=J_{n-1}$.
  \end{lem}
  \begin{proof}
    We will prove the more general result: for each $l$ there exist $f_1^l,f_2^l,\ldots,f_{n-l}^l\in I_l$ which are all simplicial and generic.
    Moreover, we will show that if $x_j$ denotes the largest variable of some $f_i^l$, then $f_i^l\in J_j$.
    Note that as $x_j\leq x_l$ then $J_j\subset I_j\subset I_l$.
 We will explicitly construct such polynomials.

    Such construction is very similar to the chordal elimination algorithm.
 The only difference is that instead of elimination ideals we use resultants.

    Initially, we assign $f^0_i = f_i$ for $1\leq i\leq n$.
 Inductively, we construct the next polynomials:
    \begin{align*}
      f^{l+1}_i = \begin{cases}
        \mathrm{Res}_{x_l}(f^l_0,f^l_{i+1}) &\mbox{ if $f^l_{i+1}$  involves $x_l$}\\
        f^{l}_{i+1} &\mbox{ if $f^l_{i+1}$ does not involve $x_l$}\\
      \end{cases}
    \end{align*}
    for $1\leq i \leq n-l$, where we assume that $f^l_0$ involves $x_l$, possibly after rearranging them.
    In the event that no $f^l_i$ involves $x_l$, then we can ignore such variable.
    Notice that Lemma~\ref{thm:resultants} tells us that $f^l_i$ are all generic and simplicial.

    We need to show that $f^l_i\in J_j$, where $x_j$ is the largest variable of $f^l_i$.
    We will prove this by induction on $l$.

    The base case is $l=0$, where $f^0_i=f_i$ are generators of $I$,
    and thus Lemma~\ref{thm:eliminiationcontainment} says that $f_i\in J_j$.
    
    Assume that the hypothesis holds for some $l$ and consider some $f:=f_i^{l+1}$. 
    Let $x_j$ be its largest variable.
    Consider first the case where $f=f^l_{i+1}$.
    By  the induction hypothesis, $f\in J_j$ and we are done.
  
    Now consider the case that $f = \mathrm{Res}_{x_l}(f_0^l,f_{i+1}^l)$.
    In this case the largest variable of both $f_0^l,f_{i+1}^l$ is $x_l$ and thus, using the induction hypothesis, both of them lie in $J_l$.
    Let $x_p$ be the parent of $x_l$.
  Using  Equation~\eqref{eq:gradeddec} we get $f\in \elim{l+1}{J_l}\subset J_p$.
  Let's see now that $x_j\leq x_p$.
  The reason is that $x_j\in X_p$, as $f\in \K[X_p]$ and $x_j$ is its largest variable.
  Thus we found an $x_p$ with $x_j\leq x_p<x_l$ and $f\in J_p$.
  If $x_j=x_p$, we are done.
  Otherwise, if $x_j<x_p$, let $x_r$ be the parent of $x_p$.
  As $f$ does not involve $x_p$, then $f\in \elim{p+1}{J_p}\subset J_r$.
  In the same way as before we get that $x_j\leq x_r<x_p$ and $f\in J_r$.
  Note that we can repeat this argument again, until we get that $f\in J_j$.
 This concludes the induction.
  \end{proof}

Lemma~\ref{thm:getsimplicial} can be used to show the domination condition and thus to certify that chordal elimination succeeds.
 In particular, we can do this in the special case that all polynomials are simplicial, as we show now.

  \begin{thm}\label{thm:exactsimplicial}
    Let $I=\langle f_1,\ldots,f_s\rangle$ be an ideal such that for each $1\leq i \leq s$, $f_i$ is  generic simplicial.
  Then chordal elimination succeeds.
  \end{thm}
  \begin{proof}
 For each $l$, let $T_l$ be the largest subtree of $T$ with minimal vertex $x_l$.
 Equivalently, $T_l$ consists of all the descendants of $x_l$.
 Let $t_l:=|T_l|$ and let $x(f_j)$ denote the largest variable of $f_j$.
 If for all $x_l$ there are at least $ t_l$ generators $f_j$ with $x(f_j)\in T_l$ then Lemma~\ref{thm:getsimplicial} implies the domination condition and we are done.
 Otherwise, let $x_l$ be the largest where this fails.
 The maximality of $x_l$ guarantees that elimination succeeds up to such point, i.e., $I_m=\elim{m}{I}$ for all $x_m\geq x_l$.
 We claim that no equation of $I_l$ involves $x_l$ and thus we can ignore it.
 Proving this claim will conclude the proof.

    If $x_l$ is a leaf of $T$, then $t_l=1$, which means that no generator of $I$ involves $x_l$.
 Otherwise, let $x_{s_1},\ldots,x_{s_r}$ be its children.
 Note that $T_l =\{x_l\} \cup T_{s_1}\cup \ldots \cup T_{s_r}$.
  We know that there are at least $t_{s_i}$ generators with $x(f_j)\in T_{s_i}$ for each $s_i$, and such bound has to be exact as $x_l$ does not have such property.
 Thus for each $s_i$ there are exactly $t_{s_i}$ generators with  $x(f_j)\in T_{s_i}$, and there is no generator with  $x(f_j)=x_l$.
 Then, for each $s_i$, when we eliminate all the $t_{s_i}$ variables in $T_{s_i}$ in the corresponding $t_{s_i}$ equations we must get the zero ideal, i.e., $\elim{s_i+1}{J_{s_i}}=0$.
 On the other hand, as there is no generator with $x(f_j)=x_l$, then all generators that involve $x_l$ are in some $T_{s_i}$.
 But we observed that the $l$-th elimination ideal in each $T_{s_i}$ is zero,  
 so that $I_l$ does not involve $x_l$, as we wanted.
\end{proof}

\section{Elimination ideals of cliques}\label{s:cliqueselim}

\begin{notation}
We will write  $I\equal J$ whenever we have $\V(I)=\V(J)$.
\end{notation}

Algorithm~\ref{alg:eliml} allows us to compute (or bound) the elimination ideals $I\cap\K[x_l,\ldots,x_{n-1}]$.
In this section we will show that once we compute such ideals, we can also compute many other elimination ideals.
 In particular, we will compute the elimination ideals of the maximal cliques of $G$.
 
 We recall the definition of the cliques $X_l$ from Equation~\eqref{eq:cliqueXl}.
 Let $H_l:=I\cap \K[X_l]$ be the corresponding elimination ideal.
 As any clique is contained in some $X_l$, we can restrict our attention to computing $H_l$.
 In particular, all maximal cliques of the graph are of the form $X_l$ for some $l$.

 The motivation behind these clique elimination ideals is to find  sparse generators of the ideal that are the closest to a Gr\"obner basis.
 Lex Gr\"obner bases can be very large, and thus finding a sparse approximation to them might be much faster as will be seen in Section~\ref{s:applications}.
 We attempt to find such ``optimal'' sparse representation by using chordal elimination.

 Specifically, let $gb_{H_l}$ denote a lex Gr\"obner basis of each $H_l$. 
 We argue that the concatenation  $\bigcup_l gb_{H_l}$ constitutes such closest sparse representation.
 In particular, the following proposition says that if there exists a lex Gr\"obner basis of $I$ that preserves the structure, then  $\bigcup_l gb_{H_l}$ is also one.

 \begin{prop}~\label{thm:gbIimpliesgbcliques}
   Let $I$ be an ideal with graph $G$ and let $gb$ be a lex Gr\"obner basis. 
   Let $H_l$ denote the clique elimination ideals, and let $gb_{H_l}$ be the corresponding lex Gr\"obner bases.
   If $gb$ preserves the graph structure, i.e., $G(gb)\subset G$, then $\bigcup_l gb_{H_l}$ is a lex Gr\"obner basis of $I$.
 \end{prop}
 \begin{proof}
   It is clear that $gb_{H_l}\subset H_l \subset I$.
   Let $m\in \initial(I)$ be some monomial, we just need to show that $m\in \initial(\bigcup_l g_{H_l})$.
   As $\initial(I)=\initial(gb)$, we can restrict $m$ to be the leading monomial $m=\lm(p)$ of some $p\in gb$.
   By the assumption on $gb$, the variables of $p$ are in some clique $X_l$ of $G$.
   Thus, $p\in H_l$ so that $m=\lm(p)\in \initial(H_l)=\initial(gb_l)$.
   This concludes the proof.
 \end{proof}

Before computing $H_l$, we will show how to obtain elimination ideals of simpler sets.
These sets are determined by the elimination tree of the graph, and we will find the corresponding elimination ideals in Section~\ref{s:lowersets}.
After that we will come back to computing the clique elimination ideals in Section~\ref{s:elimcliques}.
Finally, we will elaborate more on the relationship between lex Gr\"obner bases and clique elimination ideals in Section~\ref{s:lexgroebner}.

\subsection{Elimination ideals of lower sets}\label{s:lowersets}

We will show now how to find elimination ideals of some simple sets of the graph, which depend on the elimination tree.
 To do so, we recall that in chordal elimination we decompose $I_l = J_l + K_{l+1}$ which allows us to compute next $I_{l+1} = \elim{l+1}{J_l} + K_{l+1}$.
Observe that 
\begin{align*}
  I_{l} &= J_l + K_{l+1} \\
  &= {J_l}+ \elim{l+1}{J_l} + K_{l+1} \\
  &= {J_l}+ I_{l+1} \\
  &= {J_l}+ J_{l+1} + K_{l+2} \\
  &= {J_l}+ {J}_{l+1} +\elim{l+2}{J_{l+1}} + K_{l+2}.
\end{align*}
Continuing this way we conclude:
\begin{align}\label{eq:decomposeJ}
  I_{l}= {J_l}+ {J}_{l+1} +\cdots +{J}_{n-1}.
\end{align}
We will obtain a similar summation formula for other elimination ideals apart from $I_l$.
 
 Consider again the elimination tree $T$.
 We present another characterization of it.

\begin{prop}\label{thm:dagreduction}
  Consider the directed acyclic graph (DAG) obtained by orienting the edges of $G$ with the order of its vertices.
 Then the elimination tree $T$ corresponds to the transitive reduction of such DAG.
 Equivalently, $T$ is the Hasse diagram of the poset associated to the DAG.
\end{prop}
\begin{proof}
  As $T$ is a tree, it is reduced, and thus we just need to show that any arc from the DAG corresponds to a path of $T$.
  Let  $x_i\to x_j$ be an arc in the DAG, and observe that being an arc is equivalent to $x_j\in X_i$.
  Let $x_p$ be the parent of $x_i$.
  Then Lemma~\ref{thm:cliquecontainment} implies $x_j\in X_p$, and thus $x_p\to x_j$ is in the DAG.
  Similarly if $x_r$ is the parent of $x_p$ then $x_r\to x_j$  is another arc.
  By continuing this way we find a path $x_i,x_p,x_r,\ldots$ in $T$ that connects $x_i\to x_j$, proving that $T$ is indeed the transitive reduction.
\end{proof}

\begin{defn}
We say a set of variables $\Lambda$ is a \emph{lower set} if $T|_\Lambda$ is also a tree rooted in $x_{n-1}$.
Equivalently, $\Lambda$ is a lower set of the poset associated to the DAG of Proposition~\ref{thm:dagreduction}.
\end{defn}

Observe that $\{x_l,x_{l+1},\ldots,x_{n-1}\}$ is a lower set, as when we remove $x_0,x_1,\ldots$ we are pruning some leaf of $T$.
 The following lemma gives a simple property of these lower sets.

 \begin{lem}\label{thm:dagbranch}
  If $X$ is a set of variables such that $G|_X$ is a clique, then $T|_X$ is contained in some branch of $T$.
  In particular, if $x_l>x_m$ are adjacent, then any lower set containing $x_l$ must also contain $x_m$.
 \end{lem}
 \begin{proof}
  For the first part, note that the DAG induces a poset on the vertices, and restricted to $X$ we get a linear order. 
  Thus, in the Hasse diagram $X$ must be part of a chain (branch).
  The second part follows by considering the clique $X=\{x_l,x_m\}$ and using the previous result.
 \end{proof}

 The next lemma tells us how to obtain the elimination ideals of any lower set.

\begin{lem}\label{thm:elimlower}
  Let $I$ be an ideal, let $V=\V(I)$ and assume that the domination condition holds for chordal elimination. 
 Let $\Lambda\subset \{x_0,\ldots,x_{n-1}\}$ be a lower set.
 Then,
  \begin{align*}
    I \cap \K[\Lambda] \equal  \sum_{x_i\in \Lambda} {J}_i  
  \end{align*}
  and the corresponding variety is $\pi_{\Lambda}(V)$, where $\pi_\Lambda: \K^n \to \K^\Lambda$ is the projection onto $\Lambda$. 
\end{lem}
\begin{proof}
  See Appendix~\ref{s:proofscliqueselim}.
\end{proof}

\subsection{Cliques elimination algorithm}\label{s:elimcliques}

Lemma~\ref{thm:elimlower} tells us that we can very easily obtain the elimination ideal of any lower set.
We return now to the problem of computing the elimination ideals of the cliques $X_l$, which we denoted as $H_l$.
Before showing how to get them, we need a simple lemma.

\begin{lem}\label{thm:vertexorder}
  Let $G$ be a chordal graph and let $X$ be a clique of $G$.
 Then there is a perfect elimination ordering $v_0,\ldots,v_{n-1}$ of $G$ such that  the last vertices of the ordering correspond to $X$, i.e., $X = \{v_{n-1},v_{n-2},\ldots,v_{n-|X|}\}$ .
\end{lem}
\begin{proof}
  We can apply Maximum Cardinality Search (Algorithm~\ref{alg:mcs}) to the graph, choosing at the beginning all the vertices of clique $X$.
 As the graph is chordal, this gives a reversed perfect elimination ordering.
\end{proof}

\begin{thm}\label{thm:maxcliqueselim}
  Let $I$ be a zero dimensional ideal with chordal graph $G$.
  Assume that the domination condition holds for chordal elimination.
 Then we can further compute ideals $H_l\in \K[X_l]$ such that $H_l\equal I\cap \K[X_l]$, preserving the structure.
\end{thm}
\begin{proof}
  We will further prove that the corresponding variety is $\pi_{X_l}(V)$, where $V = \V(I)$ and $\pi_{X_l}: \K^n\to \K^{X_l}$ is the projection onto $X_l$.
  We proceed by induction on $l$.

  The base case is $l= n-1$.
  As chordal elimination is successful then $I_{n-1}\equal \elim{n-1}{I}$ and the variety is $\pi_{n-1}(V)$, so we can set $H_{n-1}= I_{n-1}$.

  Assume that we found $H_m$ for all $x_m<x_l$.
 Let $\Lambda$ be a lower set with largest element $x_l$.
 By Lemma~\ref{thm:elimlower}, we can compute an ideal $I_\Lambda$ with 
   $\V(I_\Lambda)=\pi_{\Lambda}(V)$.
 Note that $X_l\subset \Lambda$ because of Lemma~\ref{thm:dagbranch}.
 Thus, we should use as $H_l$ the ideal $I_\Lambda\cap \K[X_l]$.
 Naturally, we will use chordal elimination to approximate this ideal.
 For a reason that will be clear later, we modify $I_\Lambda$, appending to it the ideals $H_r$ for all $x_r\in \Lambda\setminus\{x_l\}$.
 Observe that this does not change the variety.
 
 Consider the induced graph $G|_\Lambda$, which is also chordal as $G$ is chordal.
 Thus, Lemma~\ref{thm:vertexorder} implies that there is a perfect elimination ordering $\sigma$ of $G|_\Lambda$ where the last clique is $X_l$.
 We can now use Algorithm~\ref{alg:eliml} in the ideal $I_\Lambda$ using such ordering of the variables to find an ideal $H_l$ that approximates $I_\Lambda\cap \K[X_l]$.
 We will show now that this elimination is successful and thus $H_l \equal I_\Lambda\cap \K[X_l]$.
 
 Let $X_j^\sigma\subset G|_\Lambda$ denote the cliques as defined in~\eqref{eq:cliqueXl} but using the new ordering $\sigma$ in $G|_\Lambda$.
 Similarly, let $I_j^\sigma=J_j^\sigma+K_{j+1}^\sigma$ denote the clique decompositions used in chordal elimination with such ordering.
 Let $x_m$ be one variable that we need to eliminate to obtain $H_l$, i.e., $x_m\in \Lambda\setminus X_l$.
 Let's assume that $x_m$ is such that $X_m^\sigma$ is a maximal clique of $G|_\Lambda$.
 As the maximal cliques do not depend on the ordering, it means that $X_m^\sigma=X_{r}$ for some $x_r<x_l$, and thus we already found an $H_r$ with $H_r\equal I\cap \K[X_m^\sigma]$.
 Observe that $H_r \subset J_m^\sigma$ by recalling that we appended $H_r$ to $ I_\Lambda$ and using Lemma~\ref{thm:eliminiationcontainment}.
 As $H_r$  is zero dimensional, then $J_m^\sigma$ is also zero dimensional  for all such $x_m$.
 Therefore, Corollary~\ref{thm:exactzerodimensional} says that the domination condition holds and chordal elimination succeeds.

 Finally, let's prove that $H_l \equal I\cap \K[X_l]$.
 Observe that as the domination condition holds in the elimination above (to get $H_l$), then
 \begin{align*}
   \V(H_l) = \V(I_\Lambda\cap \K[X_l]) = \pi_{X_l}(\V(I_\Lambda)).
 \end{align*}
 As $\V(I_\Lambda)=\pi_{\Lambda}(V)$,
 we obtain that $\V(H_l) = \pi_{X_l}(V)$.
 On the other hand, we also have $H_l \subset I\cap\K[X_l]$, so that
 \begin{align*}
   \V(H_l)\supset \V(I\cap\K[X_l]) \supset \pi_{X_l}(V).
 \end{align*}
 Therefore, the three terms above must be equal.
\end{proof}

Observe that the above proof hints to an algorithm to compute $H_l$.
However, the proof depends on the choice of some lower set $\Lambda$ for each $x_l$.
To avoid eliminations we want to use a lower set $\Lambda$ as small as possible.
By making a good choice we can greatly simplify the procedure and we get, after some observations made in Corollary~\ref{thm:elimcliques}, the Algorithm~\ref{alg:elimcliques}.
Note that this procedure recursively computes the clique elimination ideals: 
for a given node $x_l$ it only requires ${J_l}$ and the clique elimination ideal of its parent $x_p$.

\begin{algorithm}
  \caption{Compute Elimination Ideals of Cliques}
  \label{alg:elimcliques}
  \begin{algorithmic}[1]
    \Require{An ideal $I$, given by generators with chordal graph $G$}
    \Ensure{Ideals $H_l$ such that  $H_l\equal I\cap \K[X_l]$}
    \Procedure{CliquesElim}{$I,G$}
    \State get cliques ${X}_0,\ldots,{X}_{n-1} $ of $G$
    \State get ${J}_0,\ldots,{J}_{n-1} $ from \Call{ChordElim}{$I,G$}
    \State $H_{n-1} = {J}_{n-1}$
    \For {$l = n-2:0$}
    \State $x_p = $ parent of $x_l$
    \State $C = X_p\cup \{x_l\}$ 
    \State $I_C = H_p + {J}_l$
    \State $\mathit{order} = $ \Call{MCS}{$G|_{C},\mathit{start}=X_l$}
    \State $H_l = $ \Call{ChordElim}{$I_C^{\mathit{order}},G|_C^{\mathit{order}}$}
    \EndFor
    \State \Return $H_0,\ldots,H_{n-1}$
    \EndProcedure
  \end{algorithmic}
\end{algorithm}

\begin{cor}\label{thm:elimcliques}
  Let $I$ be a zero dimensional ideal with chordal graph $G$.
  Assume that the domination condition holds for chordal elimination.
  Then Algorithm~\ref{alg:elimcliques} correctly computes the clique elimination ideals, i.e., $H_l \equal I\cap\K[X_l]$.
\end{cor}
\begin{proof}
  We refer to the proof of Theorem~\ref{thm:maxcliqueselim}.
  For a given $x_l$, let $x_p$ be its parent and let $P_l$ denote the directed path in $T$ from $x_l$ to the root $x_{n-1}$.
  It is easy to see that $P_l$ is a lower set, and that $P_l = P_p \cup \{x_l\}$.
  We will see that Algorithm~\ref{alg:elimcliques} corresponds to selecting the lower set $\Lambda$ to be this $P_l$ and reusing the eliminations performed to get $H_p$ when we compute $H_l$.

  In the proof of Theorem~\ref{thm:maxcliqueselim}, to get $H_l$ we need a perfect elimination ordering (PEO) $\sigma_l$ of $G|_\Lambda$ that ends in $X_l$.
  This order $\sigma_l$ determines the eliminations performed in $I_\Lambda$.
  Let $\sigma_p$ be a PEO of $G|_{P_p}$, whose last vertices are $X_p$.
  Let's see that we can extend $\sigma_p$ to obtain the PEO $\sigma_l$ of $G|_{P_l}$.
  Let $C:= X_p\cup \{x_l\}$ and observe that $X_l\subset C$ due to Lemma~\ref{thm:cliquecontainment}, and thus $P_l = P_p \cup C$.
  Let $\sigma_C$ be a PEO of $G|_C$ whose last vertices are $X_p$ (using Lemma\ref{thm:vertexorder}).
  We will argue that the following ordering works:
  $$\sigma_l:= (\sigma_p \setminus X_p) + \sigma_C.$$

  By construction, the last vertices of $\sigma_l$ are $X_l$, so we just need to show that it is indeed a PEO of $G|_{P_l}$. 
  Let $v \in P_l$, and let $X_v^{\sigma_l}$ be the vertices adjacent to it that follow $v$ in $\sigma_l$.
  We need to show that $X_v^{\sigma_l}$ is a clique.
  There are two cases: $v \in C$ or $v\notin C$.
  If $v\in C$, then  $X_v^{\sigma_l}$ is the same as with $\sigma_C$, so that it is a clique because $\sigma_C$ is a PEO. 
  If $v\notin C$, we will see that $X_v^{\sigma_l}$ is the same as with $\sigma_p$, and thus it is a clique.
  Consider the partition $X_v^{\sigma_l} = (X_v^{\sigma_l} \setminus X_p) \cup (X_v^{\sigma_l} \cap X_p)$, and note the part that is not in $X_p$ only depends on $\sigma_p$.
  The part in $X_p$ is just $\adj(v) \cap X_p$, i.e., its neighbors in $X_p$, given that we put $\sigma_C$ to the end of $\sigma_l$.
  Observe that the same happens for $\sigma_p$, i.e., $X_v^{\sigma_p}\cap X_p=\adj(v) \cap X_p$, by construction of $\sigma_p$.
  Thus $X_v^{\sigma_l}  = X_v^{\sigma_p}$, as wanted.

  The argument above shows that given any PEO of $P_p$ and any PEO of $C$, we can combine them into a PEO of $P_l$. 
  This implies that the eliminations performed to obtain $H_p$ can be reused to obtain  $H_l$, and the remaining eliminations correspond to $G|_C$.
  Thus, we can obtain this clique elimination ideals recursively, as it is done in Algorithm~\ref{alg:elimcliques}.
\end{proof}

Computing a Gr\"obner basis for all maximal cliques in the graph might be useful as it decomposes the system of equations into simpler ones.
 We can extract the solutions of the system by solving the subsystems in each clique independently and ``glueing'' them.
 We elaborate on this now.

 \begin{lem}\label{thm:cliquesvariety}
  Let $I$ be an ideal and let $H_j=I\cap \K[X_j]$ be the cliques elimination ideals.
 Then,
 $$I= H_0 + H_{1} + \cdots + H_{n-1}.$$
\end{lem}
\begin{proof}
  As $H_j\subset I$ for any $x_j$, then $H_0+\cdots+H_{n-1} \subset I_l$.
 On the other hand, let $f\in I$ be one of its generators.
 By definition of $G$, the variables of $f$ must be contained in some $X_j$, so we have $f\in H_j$.
 This implies $I \subset H_0 +  \cdots + H_{n-1}.$
\end{proof}

 Lemma~\ref{thm:cliquesvariety} gives us a strategy to solve zero dimensional ideals.
 Note that $H_j$ is also zero dimensional. 
 Thus, we can compute the elimination ideals of the maximal cliques, we solve each $H_j$ independently, and finally we can merge the solutions.
 We illustrate that now.

 \begin{exmp}\label{exmp:3colorings}
Let $G$ be the blue/solid graph in Figure~\ref{fig:graph10notchordal}, and let $I$ be given by:
  \begin{align*}
 x_i^3 - 1&=0, &0\leq i\leq 8\\
 x_9 - 1 &= 0\\
 x_i^{2}+x_ix_j+x_j^{2}&=0,&(i,j) \mbox{ blue/solid edge}
  \end{align*}
  Note that the graph associated to the above equations is precisely $G$.
  However, to use chordal elimination we need to consider the chordal completion $\overline{G}$, which includes the three green/dashed edges of Figure~\ref{fig:graph10notchordal}.
  In such completion, we identify seven maximal cliques:
  \begin{align*}
    X_0 = \{x_0,&x_6,x_7\},
    X_1 = \{x_1,x_4,x_9\},
    X_2 = \{x_2,x_3,x_5\}\\
    X_3 &= \{x_3,x_5,x_7,x_8\},
    X_4 = \{x_4,x_5,x_8,x_9\}\\
    X_5 &= \{x_5,x_7,x_8,x_9\},
    X_6 = \{x_6,x_7,x_8,x_9\}.
  \end{align*}
  With Algorithm~\ref{alg:elimcliques} we can find the associated elimination ideals. Some of them are:
      \begin{align*}
        H_0 &= \langle x_{0} + x_{6} + 1, x_{6}^{2} + x_{6} + 1, x_{7} - 1\rangle \\
        H_5 &= \langle x_{5} - 1, x_{7} - 1, x_{8}^{2} + x_{8} + 1, x_{9} - 1\rangle \\
        H_6 &= \langle x_{6} + x_{8} + 1, x_{7} - 1, x_{8}^{2} + x_{8} + 1, x_{9} - 1\rangle.
      \end{align*}
  Denoting $\zeta = e^{2\pi i/3}$, the corresponding varieties are:
      \begin{align*}
        H_0 &: \{x_0,x_6,x_7\}  \,\,\;\;\;\; \to \left\{ \zeta,  \zeta^2, 1\right\}, \;\;\left\{ \zeta^2, \zeta, 1\right\}\\
        H_5 &: \{x_5,x_7,x_8,x_9\}\to \left\{ 1, 1,  \zeta,  1\right\},  \;  \left\{ 1, 1, \zeta^2,  1\right\} \\
        H_6 &: \{x_6,x_7,x_8,x_9\}\to \left\{ \zeta^2,  1,  \zeta, 1\right\}, \left\{\zeta, 1,\zeta^2, 1\right\}.
      \end{align*}
  There are only two solutions to the whole system, one of them corresponds to the values on the left and the other to the values on the right.
 \end{exmp}

From the example above we can see that to obtain a solution of $I$ we have to match solutions from different cliques $H_l$.
We can do this matching iteratively following the elimination tree.
Any partial solution is guaranteed to extend as the elimination was successful.
Let's see now an example were this matching gets a bit more complex.

\begin{exmp}
Consider again the blue/solid graph in Figure~\ref{fig:graph10notchordal}, and
   let $I$ be given by:
  \begin{align*}
 x_i^4 - 1&=0, &0\leq i\leq 8\\
 x_9 - 1 &= 0\\
 x_i^{3}+x_i^2x_j+x_ix_j^2+x_j^{3}&=0,&(i,j) \mbox{ blue/solid edge}
  \end{align*}
  The graph (and cliques) is the same as in Example~\ref{exmp:3colorings}, but the variety this time is larger.
  This time we have $|\V(H_0)|=18$, $|\V(H_5)|=27$, $|\V(H_6)|=12$.
  This numbers are still small.
  However, when we merge all partial solutions we obtain $|\V(I)|=528$.
\end{exmp}

\subsection{Lex Gr\"obner bases and chordal elimination}\label{s:lexgroebner}

 To finalize this section, we will show the relationship between lex Gr\"obner bases of $I$ and lex Gr\"obner bases of the clique elimination ideals $H_l$.
 We will see that both of them share many structural properties.
 This justifies our claim that these polynomials are the closest sparse representation of $I$ to a lex Gr\"obner basis.
 In some cases, the concatenation of the clique Gr\"obner bases might already be a lex Gr\"obner basis of $I$.
 This was already seen in Proposition~\ref{thm:gbIimpliesgbcliques}, and we will see now another situation where this holds.
 In other cases, a lex Gr\"obner basis can be much larger than the concatenation of the clique Gr\"obner bases.
 As we can find $H_l$ while preserving sparsity, we can outperform standard Gr\"obner bases algorithms in many cases, as will be seen in Section~\ref{s:applications}.

 We focus on radical zero dimensional ideals $I$.
 Note that this radicality assumption is not restrictive, as we have always been concerned with $\V(I)$, and we can compute $\sqrt{H_l}$ for each $l$.
 We recall now that in many cases (e.g., generic coordinates) a radical zero dimensional has a very special type of Gr\"obner bases. 
 We say that $I$ is in \emph{shape position} 
 if the reduced lex Gr\"obner basis has the structure:
 \begin{align*}
   x_0 - g_0(x_{n-1}),\,
   x_1 - g_1(x_{n-1}),\ldots,\,
   x_{n-2} - g_{n-2}(x_{n-1}),\,
   g_{n-1}(x_{n-1}).
 \end{align*}
 We will prove later the following result for ideals in shape position.
  
 \begin{prop}\label{thm:shapeposition}
   Let $I$ be a radical zero dimensional ideal in shape position. 
   Let $gb_{H_l}$ be a lex Gr\"obner basis of $H_l$.
   Then $\bigcup_{l}gb_{H_l}$ is a lex Gr\"obner basis of $I$.
 \end{prop}

 If the ideal is not in shape position, then the concatenation of such smaller Gr\"obner bases might not be a Gr\"obner basis for $I$.
 Indeed, in many cases any Gr\"obner basis for $I$ is extremely large, while the concatenated polynomials $gb_{H_l}$ are relatively small as they preserve the structure.
 This  will be seen in the application studied of Section~\ref{s:coloringsequations}, where we will show how much simpler can $\bigcup_l gb_{H_l}$ be compared to a full Gr\"obner basis.

 Even when the ideal is not in shape position, the concatenated polynomials already have some of the structure of a lex Gr\"obner basis of $I$, as we will show.
 Therefore, it is usually simpler to find such Gr\"obner basis starting from such concatenated polynomials.
 In fact, in Section~\ref{s:coloringsequations} we show that by doing this we can compute a lex Gr\"obner basis faster than a degrevlex Gr\"obner basis.

 \begin{thm}\label{thm:preservestructure}
  Let $I$ be a radical zero dimensional ideal.
  For each $x_l$ let $gb_{I_l}$ and $gb_{H_l}$ be minimal lex Gr\"obner bases for the elimination ideals $I_l=\elim{l}{I}$ and $H_l=I\cap\K[X_l]$.
  Denoting $\deg$ to the degree, the following sets are equal:
  \begin{align*}
  D_{I_l}&=\{\deg_{x_l}(p): p\in gb_{I_l}\}\\
  D_{H_l}&=\{\deg_{x_l}(p): p\in gb_{H_l}\}.
  \end{align*}
\end{thm}
\begin{proof}
  See Appendix~\ref{s:proofscliqueselim}.
\end{proof}

\begin{cor}\label{thm:initialidealstructure}
  Let $I$ be a radical zero dimensional ideal, then
  for each $x_l$ we have that $x_l^{d}\in \initial(I)$ if and only if $x_l^{d}\in \initial(H_l)$, using lex ordering.
\end{cor}
\begin{proof}
  Let $gb_{I_l}, gb_{H_l}$ be minimal lex Gr\"obner bases of $I_l,H_l$.
  As $I$ is zero dimensional then there are $d_l,d_H$ such that $x_l^{d_l}$ is the leading monomial of some polynomial in $gb_{I_l}$ and $x_l^{d_H}$ is the leading monomial of some polynomial in $gb_{H_l}$.
  All we need to show is that $d_l = d_H$.
  This follows by noting that $d_l = \max\{D_{I_l}\}$ and  $d_H = \max\{D_{H_l}\}$, following the notation from Theorem~\ref{thm:preservestructure}.
\end{proof}

 \begin{proof}[Proof of Proposition~\ref{thm:shapeposition}]
   As $I$ is in shape position, then its initial ideal has the form 
   $$\initial(I)= \langle x_0,x_1,\ldots,x_{n-2},x_{n-1}^d\rangle$$
   for some $d$.
   For each $x_l>x_{n-1}$, Corollary~\ref{thm:initialidealstructure} implies that $gb_{H_l}$ contains some $f_l$ with leading monomial $x_l$.
   For $x_{n-1}$, the corollary says that there is a $f_{n-1}\in gb_{H_{n-1}}$ with leading monomial $x_{n-1}^d$.
   Then 
   $\initial(I)=\langle \lm(f_0),\ldots,\lm(f_{n-1})\rangle$
   and as $f_l\in H_l\subset I$, these polynomials form a Gr\"obner basis of $I$.
 \end{proof}

\section{Complexity analysis}\label{s:finitefield}

Solving systems of polynomials in the general case is hard even for small treewidth, as it was shown in Example~\ref{exmp:treewidth2hard}.
Therefore, we need some additional assumptions to ensure tractable computation.
In this section we study the complexity of chordal elimination for a special type of ideals where we can prove such tractability.

Chordal elimination shares the same limitations of other elimination methods.
In particular, for zero dimensional ideals its complexity is intrinsically related to the size of the projection $|\pi_l(\V(I))|$.
Thus, we will make certain assumptions on the ideal that allow us to bound the size of this projection.
The following concept will be the key.

\begin{defn}[$q$-domination]
  We say that a polynomial $f$ is \emph{$(x_i,q)$-dominated} if its leading monomial has the form $x_i^{d}$ for some $d\leq q$.
  Let $I = \langle f_1,\ldots,f_s\rangle$, we say that $I$ is \emph{$q$-dominated} if for each $x_i$ there is a generator $f_j$ that is $(x_i,q)$-dominated.
\end{defn}

We will assume that $I$ satisfies this $q$-dominated condition.
Observe that Corollary~\ref{thm:exactsimplicialeachl} holds, and thus chordal elimination succeeds.
Note that this condition also implies that $I$ is zero dimensional.

It should be mentioned that the $q$-dominated condition applies to finite fields.
Let $\F_q$ denote the finite field of size~$q$.
If we are interested in solving a system of equations in $\F_q$ (as opposed to its algebraic closure) we can add the equations $x_i^q-x_i$.
Even more, by adding such equations we obtain the radical ideal $\I(\V_{\F_q}(I))$ \cite{gao2009counting}.

We need to know the complexity of computing a lex Gr\"obner basis.
To simplify the analysis, we assume from now that the generators of the ideal have been \emph{preprocessed} to avoid redundancies.
Specifically, we make the assumption that the polynomials have been pseudo reduced so that no two of them have the same leading monomial and no monomial is divisible by $x_i^{q+1}$.
Note that the latter assumption can be made because the ideal is $q$-dominated.
These conditions allow us to bound the number of polynomials.

\begin{lem}\label{thm:preprocessbasis}
  Let $I=\langle f_1,\ldots,f_s\rangle$ be a preprocessed $q$-dominated ideal. Then $s = O(q^n)$.
\end{lem}
\begin{proof}
  As $I$ is $q$-dominated, for each $0\leq i<n$ there is a generator $g_i$ with leading monomial $x_i^{d_i}$ with $d_i\leq q$.
  The leading monomial of all generators, other than the $g_i$'s, are not divisible by $x_i^q$.
  There are only $q^n$ monomials with degrees less that $q$ in any variable. 
  As the leading monomials of the generators are different, the result follows.
\end{proof}

The complexity of computing a Gr\"obner basis for a zero-dimensional ideal is known to be single exponential in $n$ \cite{lakshman1990complexity}. This motivates the following definition.
\begin{defn}
Let $\alpha$ be the \emph{smallest constant} such that the complexity of computing a Gr\"obner basis is $\widetilde{O}(q^{\alpha n})$ for any (preprocessed) $q$-dominated ideal. 
Here $\widetilde{O}$ ignores polynomial factors in $n$.
\end{defn}

A rough estimate of $\alpha $ is stated next. The proof in~\cite{gao2009counting} is for the case of $\F_q$, but the only property that they use is that the ideal is $q$-dominated.

\begin{prop}[\cite{gao2009counting}]\label{thm:complexitylexgroebner}
  Buchberger's algorithm  in a $q$-dominated ideal requires $O(q^{6n})$ field operations. 
\end{prop}

We should mention that the complexity of Gr\"obner bases has been actively studied and different estimates are available.
For instance, Faug\`ere et al.~\cite{faugere2013polynomial} show that for generic ideals the complexity is $\widetilde{O}(D^\omega)$, where $D$ is the number of solutions and $2<\omega<3$ is the exponent of matrix multiplication. 
Thus, if we only considered generic polynomials we could interpret such condition as saying that $\alpha \leq \omega$.
However, even if the generators of $I$ are generic, our intermediate calculations are not generic and thus we cannot make such an assumption.

Nevertheless, to obtain good bounds for chordal elimination we need a slightly stronger condition than $I$ being $q$-dominated.
Let $X_1,\ldots,X_r$ denote the \emph{maximal cliques} of the graph $G$, and let
\begin{align}\label{eq:fakeHj}
\hat{H_j} = \langle f: f\mbox{ generator of } I,\, f \in\K[X_j] \rangle.
\end{align}
Note that $\hat{H_j}\subset I\cap \K[X_j]$.
We assume that each (maximal) $\hat{H_j}$ is $q$-dominated.
Note that such condition is also satisfied in the case of finite fields.
The following lemma shows the reason why we need this assumption.

\begin{lem}\label{thm:qdominated}
  Let $I$ be such that for each maximal clique $X_j$ the ideal $\hat{H_j}$ (as in~\eqref{eq:fakeHj}) is $q$-dominated.
  Then in Algorithm~\ref{alg:eliml} we have that $J_l$ is $q$-dominated for any $x_l$.
\end{lem}
\begin{proof}
  See Appendix~\ref{s:proofsfinitefield}.
\end{proof}

It should be mentioned that whenever we have a zero dimensional ideal $I$ such that each $\hat{H_j}$ is also zero dimensional, then the same results apply by letting $q$ be the largest degree in a Gr\"obner basis of any $\hat{H_j}$.

We derive now complexity bounds, in terms of field operations, for chordal elimination under the assumptions of Lemma~\ref{thm:qdominated}.
We use the following \emph{parameters}:
$n$ is the number of variables, $s$ is the number of equations, $\kappa$ is the clique number (or treewidth), i.e., the size of the largest clique of $G$.

\begin{thm}\label{thm:complexityeliml}
  Let $I$ be such that each (maximal) $\hat{H_j}$ is $q$-dominated.
  In Algorithm~\ref{alg:eliml}, the complexity of computing $I_{l}$ is $\widetilde{O}(s+lq^{\alpha\kappa})$.
  We can find all elimination ideals  in $\widetilde{O}(nq^{\alpha\kappa})$.
Here $\widetilde{O}$ ignores polynomial factors in $\kappa$.
\end{thm}
\begin{proof}
  In each iteration there are essentially only two relevant operations: decomposing $I_l = J_l+K_{l+1}$, and finding a Gr\"obner basis for $J_l$.

  For each $x_l$, Lemma~\ref{thm:qdominated} tells us that $J_l$ is $q$-dominated.
  Thus, we can compute a lex Gr\"obner basis of $J_l$ in $\widetilde{O}(q^{\alpha\kappa})$.
  Here we assume that the initial $s$ equations were preprocessed, and note that the following equations are also preprocessed as they are obtained from minimal Gr\"obner bases.
  To obtain $I_l$ we compute at most $l$ Gr\"obner bases, which we do in $\widetilde{O}(lq^{\alpha\kappa})$.

  It just remains to bound the time of decomposing $I_l=J_l+K_{l+1}$.
  Note that if we do this decomposition in a naive way we will need  $\Theta(ls)$ operations.
  But we can improve such bound easily.
  For instance, assume that in the first iteration we compute for every generator $f_j$ the largest $x_l$ such that $f_j\in J_l$.
  Thus $f_j$ will be assigned to $K_{m+1}$ for all $x_m>x_l$, and then it will be assigned to $J_l$.
  We can do this computation in $\widetilde{O}(s)$.
  We can repeat the same process for all polynomials $p$ that we get throughout the algorithm.
  Let $s_l$ be the number of generators of $\elim{l+1}{J_l}$.
  Then we can do all decompositions in $\widetilde{O}(s+s_0+s_1+\ldots+s_{l-1})$.
  We just need to bound $s_l$.

  It follows from Lemma~\ref{thm:preprocessbasis} that for each clique $X_l$, the size of any minimal Gr\"obner basis of arbitrary polynomials in $X_l$ is at most $q^\kappa+\kappa$.
  As the number of generators of $\elim{l+1}{J_l}$ is bounded by the size of the Gr\"obner basis of $J_l\subset \K[X_l]$,
  then $s_l = \widetilde{O}(q^\kappa)$. 
  Thus, we can do all decompositions in $\widetilde{O}(s+lq^{\kappa})$.

  Thus, the total cost to compute $I_l$ is
  \begin{align*}
    \widetilde{O}(s + lq^{\kappa}+ lq^{\alpha\kappa}) = \widetilde{O}(s + lq^{\alpha\kappa}).
  \end{align*}
  In particular, we can compute $I_{n-1}$ in $\widetilde{O}(s + nq^{\alpha\kappa})$.
  Note that as each of the original $s$ equations is in some $X_l$, then Lemma~\ref{thm:preprocessbasis} implies that $s = O(nq^{\kappa})$.
  Thus, we can find all elimination ideals in $\widetilde{O}(nq^{\alpha\kappa})$.
\end{proof}
\begin{rem}
  Note that to compute the bound $W$ of~\eqref{eq:intersectionWl} we need to use chordal elimination $l$ times, so the complexity is $O(ls+l^2q^{\alpha\kappa})$. 
\end{rem}

\begin{cor}\label{thm:complexityvariety}
  Let $I$ be such that each (maximal) $\hat{H_j}$ is $q$-dominated.
  The complexity of Algorithm~\ref{alg:elimcliques} is $\widetilde{O}(nq^{\alpha\kappa})$.
  Thus, we can also describe $\V(I)$ in $\widetilde{O}(nq^{\alpha\kappa})$.
Here $\widetilde{O}$ ignores polynomial factors in $\kappa$.
\end{cor}
\begin{proof}
  The first part of the algorithm is chordal elimination, which we can do in $O(nq^{\alpha\kappa})$, as shown above.
  Observe also that Maximum Cardinality Search runs in linear time, so we can ignore it.
  The only missing part is to compute the elimination ideas of $I_C$, where $C=X_p\cup\{x_l\}$. 
  As $|C|\leq \kappa+1$, then the cost of chordal elimination is $\widetilde{O}(\kappa q^{\alpha\kappa})=\widetilde{O}(q^{\alpha\kappa})$.
  Thus the complexity of Algorithm~\ref{alg:elimcliques} is still $\widetilde{O}(nq^{\alpha\kappa})$.

  We now prove the second part. 
  As mentioned before, the Corollary~\ref{thm:exactsimplicialeachl} applies for $q$-dominated ideals, so all eliminations are successful.
  From Lemma~\ref{thm:cliquesvariety} and the following remarks we know that the elimination ideals $H_l$, found with Algorithm~\ref{alg:elimcliques}, give a natural description of $\V(I)$.
\end{proof}

The bounds above tell us that for a fixed $\kappa$, we can find all clique elimination ideals, and thus describe the variety, in $O(n)$.
This is reminiscent to many graph problems (e.g., Hamiltonian circuit, vertex colorings, vertex cover) which are NP-hard in general, but are linear for fixed treewidth~\cite{bodlaender2008combinatorial}.
Similar results hold for some types of constraint satisfaction problems~\cite{Dechter2003}.
These type of problems are said to be fixed parameter tractable (FPT) with treewidth as the parameter.

Our methods provide an algebraic solution to some classical graph problems.
In particular, we show now an application of the bounds above for finding graph colorings.
It is known that the coloring problem can be solved in linear time for bounded treewidth~\cite{bodlaender2008combinatorial}.
We can prove the same result by encoding colorings into polynomials.

\begin{cor}\label{thm:complexitycolorings}
  Let $G$ be a graph and $\bar{G}$ a chordal completion with largest clique of size~$\kappa$.
  We can describe all $q$-colorings of $G$ in $\widetilde{O}(nq^{\alpha\kappa})$.
\end{cor}
\begin{proof}
  It is known that graph $q$-colorings can be encoded with the following system of polynomials:
  \begin{subequations} \label{eq:qcolorings}
  \begin{align}
 x_i^q - 1&=0, &i\in V\\
 x_i^{q-1}+x_i^{q-2}x_j+\dots+x_ix_j^{q-2}+x_j^{q-1}&=0,&(i,j)\in E
  \end{align}
  \end{subequations}
  where $V,E$ denote the vertices and edges, and where each color corresponds to a different square root of unity~\cite{bayer1982division,hillar2008algebraic}.
  Note that the ideal $I_G$ given by these equations satisfies the $q$-dominated condition stated before.
  The chordal graph associated to such ideal is $\bar{G}$.
  The result follows from Corollary~\ref{thm:complexityvariety}.
\end{proof}

To conclude, we emphasize the differences between our results to similar methods in graph theory and constraint satisfaction.
First, note that for systems of polynomials we do not know a priori a discrete set of possible solutions.
And even if the variety is finite, the solutions may not have a rational (or radical) representation.
In addition, by using Gr\"obner bases methods we take advantage of many well studied algebraic techniques.
Finally, even though our analysis here assumes zero dimensionality, we can use our methods in underconstrained systems and, if they are close to satisfy the $q$-dominated condition, they should perform well.
Indeed, in Section~\ref{s:sensorsequations} we test our methods on underconstrained systems.

\section{Applications}\label{s:applications}

In this section we show numerical evaluations of the approach proposed in some concrete applications.
Our algorithms were implemented using Sage~\cite{sage}.
Gr\"obner bases are computed with Singular's interface~\cite{singular}, except when $\K=\F_2$ for which we use PolyBoRi's interface~\cite{polybori}.
Chordal completions of small graphs ($n<32$) are found using Sage's vertex separation algorithm.
The experiments are performed on an i7 PC with 3.40GHz, 15.6 GB RAM, running Ubuntu 12.04. 

We will show the performance of chordal elimination compared to the Gr\"obner bases algorithms from Singular and PolyBoRi.
In all the applications we give here chordal elimination is successful because of the results of Section~\ref{s:exactelim}.
It can be seen below that in all the applications our methods perform better, as the problem gets bigger, than the algorithms from Singular and PolyBoRi.

As mentioned before, chordal elimination shares some of the limitations as other elimination methods and it performs the best under the conditions studied in Section~\ref{s:finitefield}.
We show two examples that meet such conditions in Sections~\ref{s:coloringsequations} and~\ref{s:cryptoequations}.
The first case relates to the coloring problem, which was already mentioned in Corollary~\ref{thm:complexitycolorings}.
The second case is an application to cryptography, where we solve equations over the finite field $\F_2$.

After that, Sections~\ref{s:sensorsequations} and~\ref{s:differentialequations} show cases where the conditions from Section~\ref{s:finitefield} are not satisfied.
We use two of the examples from~\cite{nie2008sparse}, where they study a similar chordal approach for semidefinite programming relaxations (SDP).
Gr\"obner bases are not as fast as SDP relaxations, as they contain more information, so we use smaller scale problems.
The first example is the sensor localization problem and the second one is given by discretizations of differential equations.

\subsection{Graph colorings}\label{s:coloringsequations}

We consider Equations~\eqref{eq:qcolorings} for $q$-colorings of a graph, over the field $\K=\Q$.
We fix the graph $G$ of Figure~\ref{fig:graphunique3col} and vary the number of colors $q$.
Such graph was considered in~\cite{hillar2008algebraic} to illustrate a characterization of uniquely colorable graphs using Gr\"obner bases.
We use a different ordering of the vertices that determines a simpler chordal completion (the clique number is 5).
\begin{figure}[htb]
  \centering
  \includegraphics[scale=0.35]{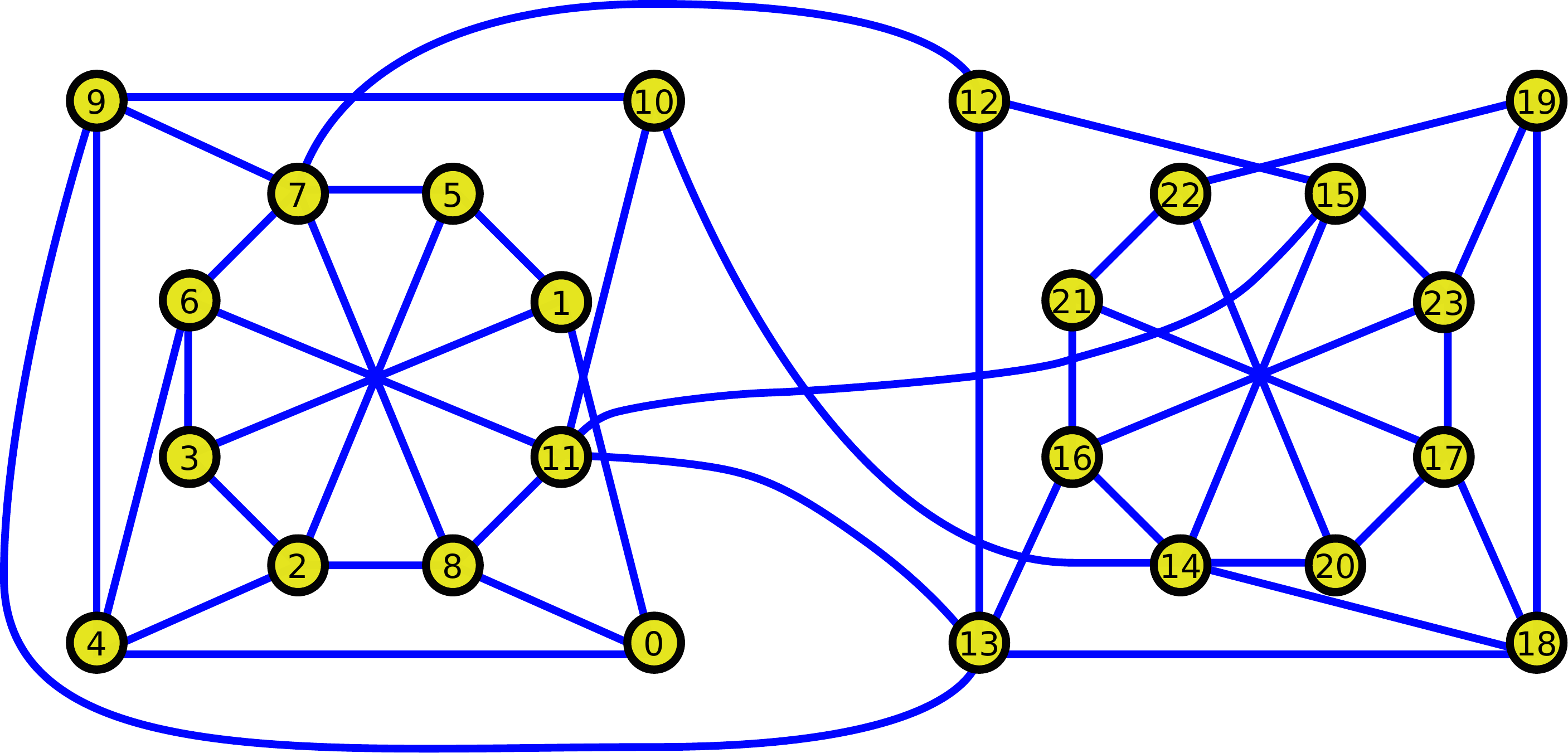}
  \caption[24-vertex graph with a unique 3-coloring]
  {Graph with a unique 3-coloring~\cite{hillar2008algebraic}.}
  \label{fig:graphunique3col}
\end{figure}

Table~\ref{tab:timingscoloringsunique} shows the performance of Algorithm~\ref{alg:eliml} and Algorithm~\ref{alg:elimcliques}, 
compared to Singular's default Gr\"obner basis algorithm using degrevlex order (lex order takes much longer).
It can be seen how the cost of finding a Gr\"obner basis increases very rapidly as we increase $q$, as opposed to our approach.
In particular, for $q=4$ we could not find a Gr\"obner basis after 60000 seconds (16.7 hours), but our algorithms ran in less than one second.
The underlying reason for such long time is the large size of the solution set (number of 4-colorings), which is $|\V(I)|=572656008$.
Therefore, it is expected that the size of its Gr\"obner basis is also very large.
On the other hand, the projection on each clique is much smaller, $|\V(H_l)|\leq 576$, and thus the corresponding Gr\"obner bases (found with Algorithm~\ref{alg:elimcliques}) are also much simpler. 

\begin{table}[htb]
\centering
\caption[Performance on $q$-colorings in a 24-vertex graph]
{Performance (in seconds) on Equations~\eqref{eq:qcolorings}  (graph of Figure~\ref{fig:graphunique3col}) for: Algorithm~\ref{alg:eliml}, Algorithm~\ref{alg:elimcliques}, computing a degrevlex Gr\"obner basis with the original equations (Singular). 
  One experiment was interrupted after 60000 seconds.
  %Singular's lex \texttt{ groebner} with input $H_l$, and Singular's degrevlex \texttt{ groebner} with the original equations.
}
  \label{tab:timingscoloringsunique}
\begin{tabular}{@{}lllllll@{}}
\toprule
$q$& Variables& Equations& Monomials&  ChordElim& CliquesElim& DegrevlexGB \\%& Lex GB from $H_l$ 
\midrule
2  & 24       & 69       & 49       &  0.058    & 0.288      & 0.001       \\%& 0.008             
\midrule
3  & 24       & 69       & 94       &  0.141    & 0.516      & 5.236       \\%& 0.010             
\midrule
4  & 24       & 69       & 139      &  0.143    & 0.615      &$>60000$     \\%&$>60000$           
\midrule
5  & 24       & 69       & 184      &  0.150    & 0.614      & -           \\%& -                 
\midrule
6  & 24       & 69       & 229      &  0.151    & 0.638      & -           \\%& -                 
\bottomrule
\end{tabular}
\end{table}

We repeat the same experiments, this time with the blue/solid graph of Figure~\ref{fig:graph10notchordal}.
Table~\ref{tab:timingscolorings} shows the results obtained.
This time we also show the cost of computing a lex Gr\"obner basis, using as input the clique elimination ideals $H_l$.
Again, we observe that chordal elimination is much faster than finding a Gr\"obner basis.
We also see that we can find faster a lex Gr\"obner basis than for degrevlex, by making use of the output from chordal elimination.

\begin{table}[htb]
\centering
  \caption[Performance on $q$-colorings in a 10-vertex graph]
  {Performance (in seconds) on Equations~\eqref{eq:qcolorings}  (blue/solid graph of Figure~\ref{fig:graph10notchordal}) for: Algorithm~\ref{alg:eliml}, Algorithm~\ref{alg:elimcliques}, 
  computing  a lex Gr\"obner basis with input $H_l$, and computing a degrevlex Gr\"obner basis with the original equations (Singular). 
  %Singular's lex \texttt{ groebner} with input $H_l$, and Singular's degrevlex \texttt{ groebner} with the original equations.
}
  \label{tab:timingscolorings}
\begin{tabular}{@{}llllllll@{}}
\toprule
$q$& Vars& Eqs& Mons& ChordElim& CliquesElim& LexGB from $H_l$ & DegrevlexGB \\ \midrule
5  & 10       & 28       & 75       & 0.035    & 0.112      & 0.003             & 0.003        \\ 
\midrule
10 & 10       & 28       & 165      & 0.044    & 0.130      & 0.064             & 0.202        \\ 
\midrule
15 & 10       & 28       & 255      & 0.065    & 0.188      & 4.539             & 8.373        \\ 
\midrule
20 & 10       & 28       & 345      & 0.115    & 0.300      & 73.225            & 105.526      \\ 
%\midrule
%25 & 10       & 28       & 435      & 0.204    & 0.496      & 480.621           & 657.809      \\ 
\bottomrule
\end{tabular}
\end{table}

\subsection{Cryptography}\label{s:cryptoequations}
We consider the parametric family $SR(n,r,c,e)$ of AES variants from~\cite{cid2005small}.
Such cypher can be embedded into a structured system of polynomials equations over $\K=\F_2$ as shown in~\cite{cid2005small}.
Note that as the field is finite the analysis from Section~\ref{s:finitefield} holds.

We compare the performance of Algorithm~\ref{alg:eliml} to PolyBoRi's default Gr\"obner bases algorithm, using both lex and degrevlex order.
As the input to the cipher is probabilistic, for the experiments we seed the pseudorandom generator in fixed values of 0, 1, 2.
We fix the values $r=1,c=2,e=4$ for the experiments
and we vary the parameter $n$, which corresponds to the number of identical blocks used for the encryption.

\begin{table}[htb]
\centering
\caption[Performance on cryptography equations]
{Performance (in seconds) on the equations of $SR(n,1,2,4)$ for: Algorithm~\ref{alg:eliml}, and computing  (lex/degrevlex) Gr\"obner bases (PolyBoRi). 
Three different experiments (seeds) are considered for each $n$.
Some experiments aborted due to insufficient memory.}
  \label{tab:timingsAES}
\begin{tabular}{@{}lllllll@{}}
\toprule
$n$                 & Variables            & Equations            & Seed &ChordElim& LexGB            & DegrevlexGB      \\ \midrule
\multirow{3}{*}{4}  & \multirow{3}{*}{120} & \multirow{3}{*}{216} & 0    & 517.018 & 217.319          & 71.223           \\ \cmidrule(l){4-7} 
                    &                      &                      & 1    & 481.052 & 315.625          & 69.574           \\ \cmidrule(l){4-7} 
                    &                      &                      & 2    & 507.451 & 248.843          & 69.733           \\ \midrule
\multirow{3}{*}{6}  & \multirow{3}{*}{176} & \multirow{3}{*}{320} & 0    & 575.516 & 402.255          & 256.253          \\ \cmidrule(l){4-7} 
                    &                      &                      & 1    & 609.529 & 284.216          & 144.316          \\ \cmidrule(l){4-7} 
                    &                      &                      & 2    & 649.408 & 258.965          & 133.367          \\ \midrule
\multirow{3}{*}{8}  & \multirow{3}{*}{232} & \multirow{3}{*}{424} & 0    & 774.067 & 1234.094         & 349.562          \\ \cmidrule(l){4-7} 
                    &                      &                      & 1    & 771.927 & $>1500$, aborted & 369.445          \\ \cmidrule(l){4-7} 
                    &                      &                      & 2    & 773.359 & 1528.899         & 357.200          \\ \midrule
\multirow{3}{*}{10} & \multirow{3}{*}{288} & \multirow{3}{*}{528} & 0    & 941.068 & $>1100$, aborted & 1279.879         \\ \cmidrule(l){4-7} 
                    &                      &                      & 1    & 784.709 & $>1400$, aborted & 1150.332         \\ \cmidrule(l){4-7} 
                    &                      &                      & 2    & 1124.942& $>3600$, aborted & $>2500$, aborted\\ \bottomrule
\end{tabular}
\end{table}

Table~\ref{tab:timingsAES} shows the results of the experiments.
We observe that for small problems standard Gr\"obner bases outperform chordal elimination, particularly using degrevlex order.
Nevertheless, chordal elimination scales better, being faster than both methods for $n=10$.
In addition, standard Gr\"obner bases have higher memory requirements, which is reflected in the many experiments that aborted for this reason.

\subsection{Sensor Network Localization}\label{s:sensorsequations}
We consider the \emph{sensor network localization} problem, also called \emph{graph realization} problem, given by the equations:
\begin{subequations} \label{eq:sensorsequations}
\begin{align}
  \|x_i - x_j\|^2 &= d_{ij}^2 &(i,j)\in \mathcal{A}\\
  \|x_i - a_k\|^2 &= e_{ik}^2 &(i,k)\in \mathcal{B}
\end{align}
\end{subequations}
where $x_1,\ldots,x_n$ are unknown sensor positions, $a_1,\ldots,a_m$ are some fixed anchors, and $\mathcal{A},\mathcal{B}$ are some sets of pairs which correspond to sensors that are close enough.
We consider the problem over the field $\K=\Q$.
Observe that the set $\mathcal{A}$ determines the graph structure of the system of equations.
Note also that the equations are simplicial (see Definition~\ref{defn:simplicial}) and thus Theorem~\ref{thm:exactsimplicial} says that chordal elimination succeeds.
However, the conditions from Section~\ref{s:finitefield} are not satisfied.

We generate random test problems in a similar way as in~\cite{nie2008sparse}.
First we generate $n=20$ random sensor locations $x_i^*$ from the unit square $[0,1]^2$.
The $m=4$ fixed anchors are $(1/2\pm 1/4, 1/2\pm 1/4)$.
We fix a proximity threshold $D$ which we set to either $D=1/4$ or $D=1/3$.
Set $\mathcal{A}$ is such that every sensor is adjacent to at most $3$ more sensors and $\|x_i-x_j\|\leq D$.
Set $\mathcal{B}$ is such that every anchor is related to all sensors with $\|x_i-a_k\|\leq D$.
For every $(i,j)\in \mathcal{A}$ and $(i,k)\in \mathcal{B}$ we compute $d_{ij},e_{ik}$.

We compare the performance of Algorithm~\ref{alg:eliml} and Singular's algorithms.
We consider Singular's default Gr\"obner bases algorithms with both degrevlex and lex orderings, and FGLM algorithm if the ideal is zero dimensional.

We use two different values for the proximity threshold $D=1/4$ and $D=1/3$.
For $D=1/4$ the system of equations is underconstrained (positive dimensional), and for $D=1/3$ the system is overconstrained (zero dimensional).
We will observe that in both cases chordal elimination performs well.
Degrevlex Gr\"obner bases perform slightly better in the overconstrained case, and poorly in the underconstrained case.
Lex Gr\"obner bases do not compete with chordal elimination in either case.

Table~\ref{tab:timingssensorsboth} summarizes the results obtained.
We used 50 random instances for the underconstrained case ($D=1/4$) and 100 for the overconstrained case ($D=1/3$).
We can see that in the underconstrained case neither lex or degrevelex Gr\"obner bases ever finished within 1000 seconds.
On the other hand, chordal elimination completes more than half of the instances.
For the overconstrained case, lex Gr\"obner basis algorithm continues to perform poorly.
On the other hand,  degrevlex Gr\"obner bases and the FGLM algorithm have slightly better statistics than chordal elimination.

\begin{table}[htb]
\centering
\caption[Performance on sensor localization equations]
{Statistics of experiments performed on random instances of Equations~\eqref{eq:sensorsequations}. We consider two situations: 50 cases of underconstrained systems ($D=1/4$) and 100 cases of overconstrained systems ($D=1/3$).
Experiments are interrupted after 1000 seconds.
}
  \label{tab:timingssensorsboth}
  %\tabcolsep=0.11cm
  %\resizebox{1.0\textwidth}{!}{
\begin{tabular}{@{}llllllll|l@{}}
\toprule
$D$                 & Repet.             & Vars              & Eqs                      &ChordElim& LexGB& DegrevlexGB& LexFGLM&              \\ \midrule
\multirow{2}{*}{1/4}&\multirow{2}{*}{ 50}&\multirow{2}{*}{40}&\multirow{2}{*}{$39\pm 5$}& 478.520 & 1000 & 1000       & -         &Mean time (s) \\ \cmidrule(l){5-9} 
                    &                    &                   &                          &  56\%   & 0\%  & 0\%        & -         &Completed     \\ \midrule
\multirow{2}{*}{1/3}&\multirow{2}{*}{100}&\multirow{2}{*}{40}&\multirow{2}{*}{$48\pm 6$}&  298.686& 1000 & 219.622    & 253.565   &Mean time (s) \\ \cmidrule(l){5-9} 
                    &                    &                   &                          &  73\%   & 0\%  & 81\%       & 77\%      &Completed     \\ \bottomrule
\end{tabular}
\end{table}

Despite the better statistics of degrevlex and FGLM in the overconstrained case, one can identify that for several of such instances chordal elimination performs much better.
This can be seen in Figure~\ref{fig:histogram_sensors}, where we observe the scatter plot of the performance of both FGLM and Algorithm~\ref{alg:eliml}.
In about half of the cases (48) both algorithms are within one second and for the rest: in 29 cases FGLM is better, in 23 chordal elimination is better.
To understand the difference between these two groups, we can look at the clique number of the chordal completions.
Indeed, the 23 cases where chordal elimination is better have a mean clique number of $5.48$, compared to $6.97$ of the 29 cases where FGLM was better.
This confirms that chordal elimination is a suitable method for cases with chordal structure, even in the overconstrained case.

\begin{figure}[htb]
  \centering
  \includegraphics[width=9cm]{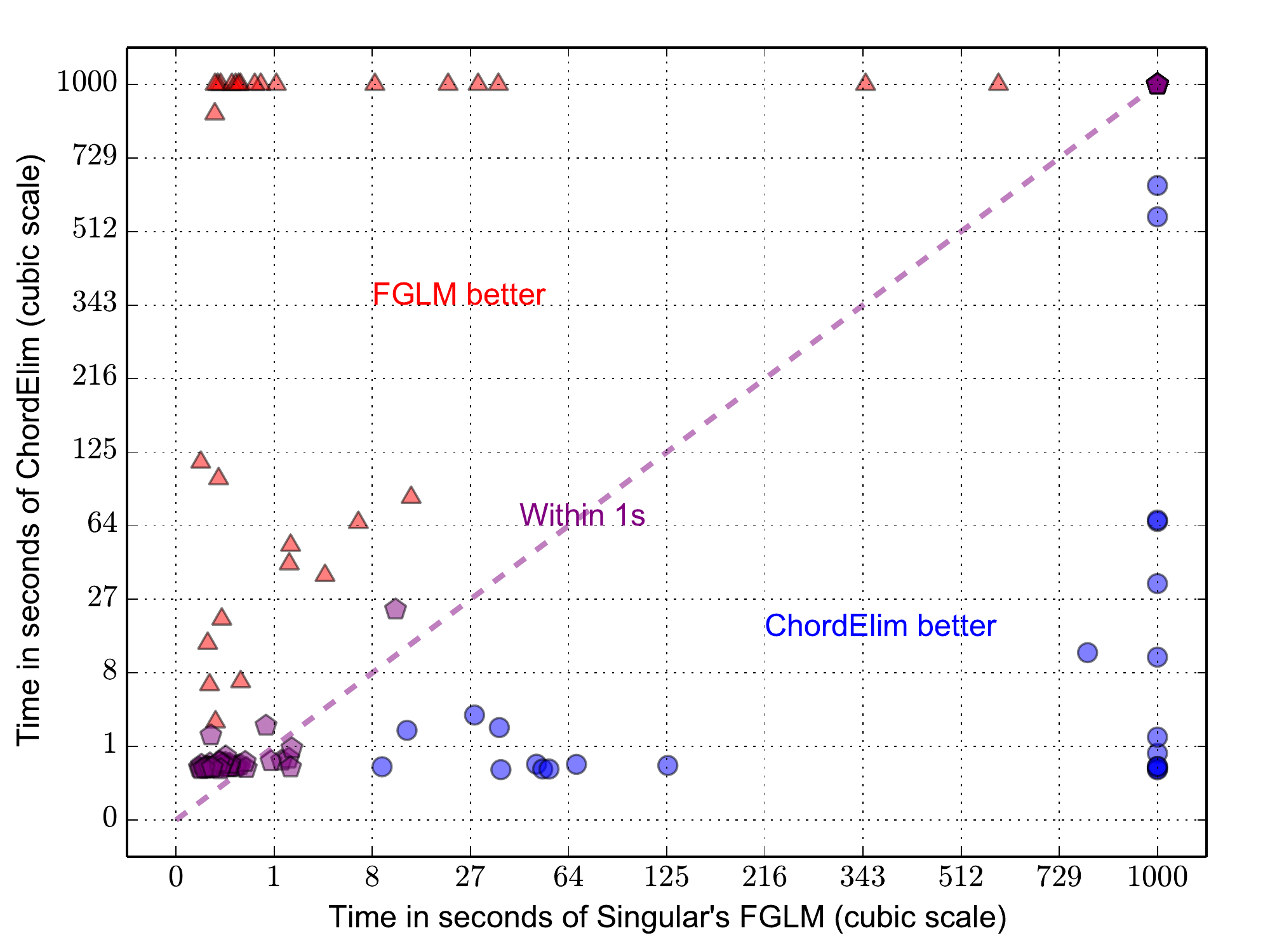}
  \caption[Scatter plot of performance on overconstrained sensors equations]
  {Scatter plot of the time used by Singular's FGLM and Algorithm~\ref{alg:eliml} on 100 random overconstrained ($D=1/3$) instances of Equations~\eqref{eq:sensorsequations}. Darker points indicate overlap.} 
  \label{fig:histogram_sensors}
\end{figure}

\subsection{Differential Equations}\label{s:differentialequations}
We consider now the following equations over the field $\K=\Q$:
\begin{subequations} \label{eq:differentialequations}
\begin{align}
0&=2x_1 - x_2 + \frac{1}{2}h^2(x_1+t_1)^3\\
0&=2x_i - x_{i-1} - x_{i+1} + \frac{1}{2}h^2(x_i+t_i)^3&\mbox{ for }i=2,\ldots,n-1\\
0&=2x_n - x_{n-1} + \frac{1}{2}h^2(x_n+t_n)^3
\end{align}
\end{subequations}
with $h = 1/(n+1)$ and $t_i = i/(n+1)$. 
Such equations were used in~\cite{nie2008sparse}, and arise from discretizing the following differential equation with boundary conditions:
\begin{align*}
x'' + \frac{1}{2}(x+t)^3 = 0, \;\;\;\;\;\;\;\; x(0)=x(1)=0.
\end{align*}
Note that these polynomials are simplicial (see Definition~\ref{defn:simplicial}) and thus chordal elimination succeeds because of Theorem~\ref{thm:exactsimplicial}.
Even more, the equations $J_l$ obtained in chordal elimination form a lex Gr\"obner basis.
However, the results from Section~\ref{s:finitefield} do not hold.
Nevertheless, we compare the performance of chordal elimination with Singular's default Gr\"obner basis algorithm with lex order.
We also consider Singular's FGLM implementation.

\begin{table}[htb]
\centering
\caption[Performance on finite difference equations]
{Performance (in seconds) on Equations~\eqref{eq:differentialequations} for: Algorithm~\ref{alg:eliml}, and computing  a lex Gr\"obner basis with two standard methods (Singular's default and FGLM). 
}
  \label{tab:timingsdifferential}
\begin{tabular}{@{}llllll@{}}
\toprule
$n$& Variables & Equations & ChordElim   & LexGB    & LexFGLM \\ \midrule
3  & 3         & 3         & 0.008       & 0.003    & 0.007      \\ \midrule
4  & 4         & 4         & 0.049       & 0.044    & 0.216      \\ \midrule
5  & 5         & 5         & 1.373       & 1.583    & 8.626      \\ \midrule
6  & 6         & 6         & 76.553      & 91.155   & 737.989    \\ \midrule
7  & 7         & 7         & 7858.926    & 12298.636& 43241.926  \\ \bottomrule
\end{tabular}
\end{table}

Table~\ref{tab:timingsdifferential} shows the results of the experiments.
The fast increase in the timings observed is common to all methods.
Nevertheless, it can be seen that chordal elimination performs faster and scales better than standard Gr\"obner bases algorithms.
Even though the degrevlex term order is much simpler in this case, FGLM algorithm is not efficient to obtain a lex Gr\"obner basis.

%\nocite {baz,fuzz,bong}
\bibliography{groebner}

\bibliographystyle{plain}

\appendix

\section{Additional proofs}\label{s:appendixproofs}

\subsection{Proofs from Section~\ref{s:chordelim}}\label{s:proofschordelim}

\begin{proof}[Proof of Theorem~\ref{thm:eliml}]
  We prove it by induction on $l$.
 The base case is Lemma~\ref{thm:elim}.
 Assume that the result holds for some $l$ and let's show it for $l+1$.

  By induction hypothesis $I_l,W_1,\ldots,W_l$ satisfy Equation~\eqref{eq:eliml1} and Equation~\eqref{eq:eliml2}.
 Lemma~\ref{thm:elim} with $I_l$ as input tell us that $I_{l+1},W_{l+1}$ satisfy:
\begin{align}
  \overline{\pi(\V(I_l))} &\subset \V(I_{l+1})\\
  \V(I_{l+1}) - \V(W_{l+1}) &\subset \pi(\V(I_l))
\end{align}
  where $\pi:\K^{n-l}\to\K^{n-l-1}$ is the projection onto the last factor.
  Then,
  \begin{align*}
    \pi_{l+1}(\V(I)) =\pi(\pi_{l}(\V(I))) \subset \pi(\V(I_l)) \subset \V(I_{l+1})
  \end{align*}
  and as $\V(I_{l+1})$ is closed, we can take the closure.
 This shows Equation~\eqref{eq:eliml1}.

  We also have
  \begin{align*}
    \pi_{l+1}(\V(I)) = \pi(\pi_l(\V(I))) &\supset \pi(\V(I_l) - [\pi_l(\V(W_1))\cup \cdots \cup \pi_l(\V(W_l))]) \\
    &\supset \pi(\V(I_l)) -\pi[\pi_l(\V(W_1))\cup \cdots \cup \pi_l(\V(W_l))] \\
    &= \pi(\V(I_l)) -[\pi_{l+1}(\V(W_1))\cup \cdots \cup \pi_{l+1}(\V(W_l))] \\
    &\supset (\V(I_{l+1})-\V(W_{l+1})) -[\pi_{l+1}(\V(W_1))\cup \cdots \cup \pi_{l+1}(\V(W_l))] \\
    &= \V(I_{l+1}) -[\pi_{l+1}(\V(W_1))\cup \cdots \cup \pi_{l+1}(\V(W_{l+1}))]
  \end{align*}
  which proves Equation~\eqref{eq:eliml2}.
\end{proof}

\subsection{Proofs from Section~\ref{s:cliqueselim}}\label{s:proofscliqueselim}

\begin{proof}[Proof of Lemma~\ref{thm:elimlower}]
  Let $H_\Lambda:=I \cap \K[\Lambda]$ and ${J}_\Lambda := \sum_{x_i\in \Lambda} {J}_i$.
  Let $x_l\in \Lambda$ be its largest element.
  For a fixed $x_l$, we will show by induction on $|\Lambda|$ that $\V(H_\Lambda) = \V(J_\Lambda) = \pi_{\Lambda}(V)$.

  The base case is when $\Lambda=\{x_l,\ldots, x_{n-1}\}$. 
  Note that as $x_l$ is fixed, such $\Lambda$ is indeed the largest possible lower set.
  In such case ${J}_\Lambda=I_l$ as seen in Equation~\eqref{eq:decomposeJ}, and as we are assuming that the domination condition holds, then $\V(H_\Lambda) = \V(I_l) = \pi_l(V)$.
 
  Assume that the result holds for $k+1$ and let's show it for some $\Lambda$ with $|\Lambda|=k$.
 Consider the subtree $T_l = T|_{\{x_l,\ldots,x_{n-1}\}}$ of $T$.
  As $T_l|_\Lambda$ is a proper subtree of $T_l$ with the same root, there must be an $x_m < x_l$ with $x_m\notin \Lambda$ and such that $x_m$ is a leaf in $T_l|_{\Lambda'}$, where $\Lambda' = \Lambda\cup \{x_m\}$.
  We apply the induction hypothesis in $\Lambda'$, obtaining that $\V(H_{\Lambda'})= \V(J_{\Lambda'})= \pi_{\Lambda'}(V)$.

 Note now that $J_m$ is is a subset of both $H_{\Lambda'},J_{\Lambda'}$.
 Observe also that we want to eliminate $x_m$ from these ideals to obtain $H_\Lambda, J_\Lambda$.
 To do so, let's change the term order to $x_m>x_l>x_{l+1}>\cdots>x_{n-1}$. 
 Note that such change has no effect inside $X_m$, and thus the term ordering for ${J_m}$ remains the same.
 As the domination condition holds, then $J_m$ is $x_m$-dominated, and thus $H_{\Lambda'},J_{\Lambda'}$ are also $x_m$-dominated.
 This means that  Lemma~\ref{thm:exactelim} holds for $H_{\Lambda'},J_{\Lambda'}$ when we eliminate $x_m$, and then
 \begin{align*}
   \V(\elim{m+1}{H_{\Lambda'}}) &= \pi_{m+1}(\V(H_\Lambda')) = \pi_{\Lambda}(V)\\
   \V(\elim{m+1}{J_{\Lambda'}}) &= \pi_{m+1}(\V(J_\Lambda')) = \pi_{\Lambda}(V).
 \end{align*}

 Notice that $H_\Lambda = \elim{m+1}{H_{\Lambda'}}$, so all we have to do now is to show that  $J_\Lambda \equal \elim{m+1}{J_{\Lambda'}}$.
 Note that
 \begin{align*}%\label{eq:approximate_J_m}
   \elim{m+1}{{J}_{\Lambda'}}= \elim{m+1}{{J}_m+{J}_\Lambda}.
 \end{align*}
 Observe that the last expression is reminiscent of Lemma~\ref{thm:elim}, but in this case we are eliminating $x_m$.
 As mentioned before, $J_m$ is $x_m$-dominated, so elimination succeeds.
 Therefore, we have
 \begin{align*}
   \elim{m+1}{{J}_m+{J}_\Lambda}\equal \elim{m+1}{{J}_m} + {J}_\Lambda.
 \end{align*}
 Let $x_p$ be the parent of $x_m$ in $T$.
 Then Equation~\eqref{eq:gradeddec} says that $ \elim{m+1}{{J}_m}\subset J_p$, where we are using that the term order change maintains $J_m$.
 Observe that $x_p\in \Lambda$ by the construction of $x_m$, and then $J_p\subset J_\Lambda$.
 Then, 
 \begin{align*}%\label{eq:approximateJm2}
   \elim{m+1}{{J}_m} + {J}_\Lambda =  {J}_\Lambda.
 \end{align*}
 Combining the last three equations we complete the proof.
\end{proof}

For the proof of Theorem~\ref{thm:preservestructure}, we will need two results before.

\begin{lem}\label{thm:cliquesvarietyIl}
  Let $I$ be a zero dimensional ideal, let $H_j=I\cap \K[X_j]$ and let $I_l = \elim{l}{I}$.
  Then,
 $$I_l \equal H_l + H_{l+1} + \cdots + H_{n-1}.$$
\end{lem}
\begin{proof}
  For each $x_j$ let $gb_{H_j}$ be a lex Gr\"obner basis of $H_j$.
  Let $F = \bigcup_{x_j} gb_{H_j}$ be the concatenation of all $gb_{H_j}$'s.
  Then the decomposition of $I$ from Lemma~\ref{thm:cliquesvariety} says that $I=\langle F \rangle$.
  Observe now that if we use chordal elimination on $F$, at each step we only remove the polynomials involving some variable;
  we never generate a new polynomial.
  Therefore our approximation of the $l$-th elimination ideal is given by $F_l = \bigcup_{x_j\leq x_l} gb_{H_j}$.
  Note now that as $H_j$ is zero dimensional it is also $x_j$-dominated, and thus Corollary~\ref{thm:exactsimplicialeachl} says that elimination succeeds.
  Thus $I_l \equal \langle F_l \rangle = \sum_{x_j\leq x_l} H_j$.
\end{proof}

 \begin{thm}[\cite{gao2003grobner}]\label{thm:lexgrobstructure}
   Let $I$ be a radical zero dimensional ideal and $V=\V(I)$.
   Let $gb$ be a minimal Gr\"obner basis with respect to an elimination order for $x_0$. 
   Then the set 
   $$D=\{\deg_{x_0}(p): p\in gb\}$$  
   where $\deg$ denotes the degree, is the same as 
   $$F=\{|\pi^{-1}(z)\cap V|: z \in \pi(V)\}$$
   where $\pi:\K^n\to \K^{n-1}$ is the projection eliminating $x_0$.
 \end{thm}

\begin{proof}[Proof of Theorem~\ref{thm:preservestructure}]
  If $x_l=x_{n-1}$, then $I_l=H_l$ and the assertion holds.
  Otherwise, note that $I_l,H_l$ are also radical zero dimensional so we can use Theorem~\ref{thm:lexgrobstructure}. 
  Let
  \begin{align*}
F_{I_l}&=\{|\pi_{I_l}^{-1}(z)\cap\V(I_l)|: z \in \pi_{I_l}(\V(I_l))\}\\
F_{H_l}&=\{|\pi_{H_l}^{-1}(z)\cap\V(H_l)|: z \in \pi_{H_l}(\V(H_l))\}
  \end{align*}
  where  $\pi_{I_l}:\K^{n-l}\to \K^{n-l-1}$ and $\pi_{H_l}:\K^{|X_l|}\to \K^{|X_l|-1}$ are projections eliminating $x_l$.
  Then we know that $D_{I_l} = F_{I_l}$ and $D_{H_l}=F_{H_l}$, so we need to show that $F_{I_l}=F_{H_l}$.

  For some $z\in \K^{n-l}$, let's denote $z=:(z_l,z_H,z_I)$ where $z_l$ is the $x_l$ coordinate, $z_H$ are the coordinates of $X_l\setminus x_l$, and $z_I$ are the coordinates of $\{x_{l},\ldots,x_{n-1}\}\setminus X_l$.
  Thus, we have $\pi_{I_l}(z)= (z_H,z_I)$ and $\pi_{H_l}(z_l,z_H)=z_H$.

  As $I$ is zero dimensional, then Lemma~\ref{thm:cliquesvarietyIl} implies that $I_l \equal H_l + I_{l+1}$.
  Note also that $\V(I_{l+1})=\pi_{I_l}(\V(I_l))$ as it is zero dimensional.
  Then,
  \begin{align*}
    z \in \V(I_l) \iff (z_l,z_H)\in \V(H_l) \mbox{ and } (z_H,z_I)\in \pi_{I_l}(\V(I_l)).
  \end{align*}
  Thus, for any $(z_H,z_I)\in \pi_{I_l}(\V(I_l))$ we have 
  \begin{align*}
    (z_l,z_H,z_I) \in \V(I_l) \iff (z_l,z_H)\in \V(H_l).
  \end{align*}
  Equivalently, for any $(z_H,z_I)\in \pi_{I_l}(\V(I_l))$ we have
  \begin{align}\label{eq:fibercontainment}
    z \in \pi_{I_l}^{-1}(z_H,z_I)\cap \V(I_l) \iff \rho(z) \in \pi_{H_l}^{-1}(z_H)\cap \V(H_l)
  \end{align}
  where $\rho(z_l,z_H,z_I):=(z_l,z_H)$.
  Therefore, $F_{I_l}\subset F_{H_l}$.
  
  On the other hand, note that if $z_H \in  \pi_{H_l}(\V(H_l))$, then there is some $z_I$ such that $(z_H,z_I)\in  \pi_{I_l}(\V(I_l))$. 
  Thus, for any $z_H \in\pi_{H_l}(\V(H_l))$ there is some $z_I$ such that \eqref{eq:fibercontainment} holds.
  This says that $F_{H_l}\subset F_{I_l}$, completing the proof.
\end{proof}

\subsection{Proofs from Section~\ref{s:finitefield}}\label{s:proofsfinitefield}

\begin{proof}[Proof of Lemma~\ref{thm:qdominated}]
  Let $x_l$ be arbitrary and let $x_m \in X_l$.
  We want to find a generator of $J_l$ that is $(x_m,q)$-dominated.
  Let $x_j\geq x_l$ be such that $X_l\subset X_j$ and $X_j$ is a maximal clique.
  Note that $x_m \in X_j$.
  Observe that  $\hat{H_j}\subset J_j$ because of Lemma~\ref{thm:eliminiationcontainment}, and thus $J_j$ is $q$-dominated. 
  Then there must be a generator $f\in J_j$ that is $(x_m,q)$-dominated.

  Let's see that $f$ is a generator of $J_l$, which would complete the proof.
  To prove this we will show that $f\in \K[X_l]$, and then the result follows from Lemma~\ref{thm:eliminiationcontainment}.
  As the largest variable of $f$ is $x_m$, then all its variables are in 
  $X_j \setminus \{x_{m+1},\ldots,x_{j}\} \subset X_j \setminus \{x_{l+1},\ldots,x_{j}\}.$
  Thus, it is enough to show that 
  $$X_j \setminus \{x_{l+1},x_{l+2},\ldots,x_{j}\}\subset X_l.$$

  The equation above follows by iterated application of Lemma~\ref{thm:cliquecontainment}, as we will see.
  Let $x_p$ be the parent of $x_j$ in $T$, and observe that $x_l\in X_p$ as $x_l\leq x_p$ and both are in clique $X_j$.
  Then Lemma~\ref{thm:cliquecontainment} implies that $X_j\setminus \{x_{p+1},\ldots,x_{j}\}\subset X_p$.
  If $x_p = x_l$, we are done.
  Otherwise, let $x_r$ be the parent of $x_p$, and observe that $x_l \in X_r$ as before.
  Then,
  $$X_j\setminus \{x_{r+1},\ldots,x_{j}\}\subset X_p\setminus \{x_{r+1},\ldots,x_{p}\}\subset X_r.$$
  If $x_r = x_l$, we are done.
  Otherwise, we can continue this process that has to eventually terminate.
  This completes the proof.
\end{proof}

\end{document}